\documentclass[11pt,letterpaper]{amsart}
\usepackage{fullpage}
\usepackage{amsmath,amssymb,amsthm,mathscinet,eucal}
\usepackage{enumitem}
\usepackage[table]{xcolor}
\usepackage{cite, hyperref}
\usepackage[noabbrev,capitalize]{cleveref}

\renewcommand{\epsilon}{\varepsilon}
\renewcommand{\phi}{\varphi}
\DeclareMathOperator{\abs}{abs}
\DeclareMathOperator{\As}{As}
\DeclareMathOperator{\CAs}{CAs}
\DeclareMathOperator{\CDF}{CDF}
\DeclareMathOperator{\CCon}{CCon}
\DeclareMathOperator{\Coh}{Coh}
\DeclareMathOperator{\CSat}{CSat}
\DeclareMathOperator{\GELO}{GELO}
\DeclareMathOperator{\id}{id}
\DeclareMathOperator{\mono}{mono}
\DeclareMathOperator{\CIsom}{CIsom}
\DeclareMathOperator{\sgn}{sgn}
\DeclareMathOperator{\CSols}{CSols}
\DeclareMathOperator{\Pair}{Pair}
\DeclareMathOperator{\Part}{Part}
\DeclareMathOperator{\SSCC}{SSCC}
\DeclareMathOperator{\Stab}{Stab}
\newcommand{\C}{{\mathbb C}}
\newcommand{\E}{{\mathbb E}}
\newcommand{\N}{{\mathbb N}}
\newcommand{\Q}{{\mathbb Q}}
\newcommand{\Z}{{\mathbb Z}}
\newcommand{\cP}{{\mathcal P}}
\newcommand{\cQ}{{\mathcal Q}}
\newcommand{\cR}{{\mathcal R}}
\newcommand{\cS}{{\mathcal S}}
\newcommand{\cW}{{\mathcal W}}
\newcommand{\fC}{{\mathfrak C}}
\newcommand{\fP}{{\mathfrak P}}
\newcommand{\fS}{{\mathfrak S}}
\newcommand{\card}[1]{\left|{#1}\right|}
\newcommand{\conj}[1]{\overline{#1}}
\newcommand{\sums}[1]{\sum_{\substack{#1}}}
\newcommand{\ceil}[1]{\lceil{#1}\rceil}
\newcommand{\ev}{{\E^\ell_{(f,g)}}}
\newcommand{\evh}{{\E^\ell_h}}
\newcommand{\cmomv}[1]{{\mu^\ell_{{#1},(f,g)}}}
\newcommand{\cmomhv}[1]{{\mu^\ell_{{#1},h}}}
\newcommand{\cmom}{\cmomv{p}}
\newcommand{\cmomh}{\cmomhv{p}}
\newcommand{\var}[1]{\cmomv{2}}
\newcommand{\smom}{{\tilde{\mu}^\ell_{p,(f,g)}}}
\newcommand{\Wrp}{V^{(p)}}
\newcommand{\Grp}{\Gamma^{(p)}}
\newcommand{\Wrtwo}{V^{(2)}}
\newcommand{\Grtwo}{\Gamma^{(2)}}
\newcommand{\Wrthree}{V^{(3)}}
\newcommand{\Grthree}{\Gamma^{(3)}}
\newcommand{\Ccong}{\cong_C}
\newcommand{\indexset}{{[p]\times[2]\times[2]}}
\newcommand{\eindexset}{{E\times[2]\times[2]}}
\newcommand{\leindexset}{{E\times[2]\times\{0\}}}
\newcommand{\reindexset}{{E\times[2]\times\{1\}}}
\newcommand{\bindexset}{{[2]\times [2]}}
\newtheorem{theorem}{Theorem}[section]
\newtheorem{proposition}[theorem]{Proposition}
\newtheorem{lemma}[theorem]{Lemma}
\newtheorem{corollary}[theorem]{Corollary}
\theoremstyle{definition}
\newtheorem{definition}[theorem]{Definition}
\newtheorem{example}[theorem]{Example}

\allowdisplaybreaks

\begin{document}
\title[Moments of crosscorrelation demerit factors]{Moments of crosscorrelation demerit factors of binary sequences}
\author{Daniel J.~Katz and Harmony M.~Vargas}\thanks{Daniel J.~Katz is and Harmony M.~Vargas was with the Department of Mathematics, California State University, Northridge.  Harmony M.~Vargas is with the Department of Mathematics, University of Nebraska-Lincoln.  This paper is based upon work supported in part by the National Science Foundation under Grant 2206454.}
\date{16 November 2026}

\begin{abstract}
Families of sequences with low mutual aperiodic crosscorrelation assist the design of systems for multi-user asynchronous communications and multiple-input multiple-output radar.
The crosscorrelation demerit factor of a pair of sequences is the sum of the squared magnitudes of their crosscorrelation values at every shift when the sequences are normalized to unit Euclidean norm, and the merit factor is the reciprocal of the demerit factor.
For each positive integer $\ell$, we endow the $2^{2 \ell}$ pairs of binary sequences of length $\ell$ with uniform probability measure and study the distribution of their crosscorrelation demerit factors.
Sarwate showed that the mean value is always $1$ regardless of length $\ell$.
We develop a method for finding an exact formula for the $p$th central moment (for any positive integer $p$) as a function of $\ell$.
Formulae for the variance and third central moment ($p=2$ and $3$) are then obtained by hand calculations, while the fourth through sixth central moments are obtained by computer-assisted calculations.
Our theory also shows that all the central moments must be strictly positive for $p\geq 2$ and $\ell \geq 3$.  
\end{abstract}

\maketitle

\section{Introduction}
In this paper we consider binary sequences aperiodically.
It is convenient for us to represent an aperiodic sequence as a doubly infinite list $f=(\ldots,f_{-1},f_0,f_1,f_2,\ldots)$ of complex numbers with only finitely many nonzero terms.
If $\ell$ is a nonnegative integer, then this $f$ is a {\it binary sequence of length $\ell$} if and only if the terms $f_0,f_1,\ldots,f_{\ell-1}$ lie in $\{-1,1\}$ and all other terms equal $0$.
Binary sequences are used extensively in communications, design of scientific instruments, and remote sensing \cite{Golomb,Golomb-Gong,Schroeder}.
In multi-user communications networks, it is important to have families of sequences such that no sequence in the family resembles a time-delayed version of any other sequence in the family.

One measure of resemblance between two sequences is aperiodic crosscorrelation.
The aperiodic crosscorrelation function for a pair of sequences determines both their periodic (even) and negaperiodic (odd) crosscorrelation functions, which are critical for performance analysis in phase-modulated communications systems.
If $f$ and $g$ are sequences and $s \in \Z$, then the {\it aperiodic crosscorrelation of $f$ with $g$ at shift $s$} is
\[
C_{f,g}(s) = \sum_{j \in \Z} f_{j+s} \conj{g_j}.
\]
Since both $f$ and $g$ have only finitely many nonzero terms, this is essentially a finite sum, and the sum can be nonzero for only finitely many shifts $s$.
One should note that $C_{g,f}(s)=\conj{C_{f,g}(-s)}$, so that the crosscorrelation function of $g$ with $f$ is deducible from that of $f$ with $g$.
When $f=g$, the value $C_{f,f}(s)$ is the {\it aperiodic autocorrelation of $f$ at shift $s$}, which is an important quantity for determining the performance of $f$ in systems requiring synchronization.
For example, in our communications network we would want $|C_{f,f}(s)|$ for each nonzero $s$ to be small compared with $C_{f,f}(0)$, so that the strong peak in the autocorrelation indicates synchronization.
Note that $C_{f,f}(0)$ is the squared Euclidean norm of $f$.
For two distinct sequences $f$ and $g$ in our network, we would like $|C_{f,g}(s)|$ to be small for all values of $s$, so that two users may be distinguished from each other regardless of the delay between their signals modulated by $f$ and $g$.

There are two measures of smallness of the correlation values for a pair $(f,g)$ of sequences.
We call worst-case ($\ell^\infty$) measure {\it peak crosscorrelation}: it is the maximum of $|C_{f,g}(s)|$ over all $s \in \Z$.
The other is a mean square ($\ell^2$) measure called the {\it crosscorrelation demerit factor of $f$ and $g$} (where $f$ and $g$ are nonzero sequences), defined as
\begin{equation}\label{Barbara}
\CDF(f,g)=\frac{\sum_{s \in \Z} |C_{f,g}(s)|^2}{C_{f,f}(0) C_{g,g}(0)},
\end{equation}
and we should note that $\CDF(g,f)=\CDF(f,g)$ because $C_{g,f}(s)=\conj{C_{f,g}(-s)}$.
The terms $C_{f,f}(0)$ and $C_{g,g}(0)$ in the denominator of our definition make the crosscorrelation demerit factor the same as the sum over all shifts of the squared magnitudes of the crosscorrelation values of the sequences obtained by normalizing $f$ and $g$ to have Euclidean norm $1$.
The performance of a phase-coded spread-spectrum multiple-access communications system is largely determined by sums of squared magnitudes of crosscorrelations of the sequences used to modulate the signals \cite{Pursley,Karkkainen}
The crosscorrelation demerit factor is a natural normalization of this sum of squared magnitudes, since it is invariant under scalar multiplication of the sequences.

For each length $\ell$, we would like to understand the distribution of the crosscorrelation demerit factors for pairs of binary sequences of that length.
Since $C_{f,f}(0)=\ell$ for every binary sequence of length $\ell$, our denominator in \eqref{Barbara} is the constant $\ell^2$ for all pairs $(f,g)$ of such sequences.
Thus, it is the numerator of \eqref{Barbara} that is the interesting and difficult part to study, and so we define
\[
\SSCC(f,g) = \sum_{s \in \Z} |C_{f,g}(s)|^2,
\]
which is the sum of the squared magnitudes of all crosscorrelation values for $f$ with $g$, and then we have
\begin{equation}\label{Celeste}
\CDF(f,g) = \frac{\SSCC(f,g)}{\ell^2}.
\end{equation}

For this entire paper $\Pair(\ell)$ denotes the set of all $2^{2\ell}$ pairs of binary sequences of length $\ell$ with uniform probability measure.
The expected value of a random variable $v$ with respect to this distribution is denoted by $\ev v(f,g)={\mathbf E}_{(f,g) \in \Pair(\ell)}(v(f))$.
With respect to this distribution, our random variable's $p$th central moment is denoted
\[
\cmom v(f,g) = \ev\left(v(f,g) - \ev v(f,g)\right)^p
\]
and the $p$th standardized moment is denoted
\[
\smom v(f,g) = \frac{\cmom v(f,g)}{\left(\var v(f,g)\right)^{p/2}}.
\]
The mean value of $\SSCC(f,g)$ for $(f,g) \in \Pair(\ell)$  was found by Sarwate \cite[eq.~(38)]{Sarwate} to be $\ell^2$; we state this result in terms of the crosscorrelation demerit factor.
\begin{theorem}[Sarwate, 1984]\label{Abigail}
If $\ell$ is a positive integer, then
\[
\E^\ell_{(f,g)} \CDF(f,g)=1.
\]
\end{theorem}  
Notice that the mean value is independent of length.
In the same paper, Sarwate also determined the mean square value for the autocorrelation demerit factor (which equals $\CDF(f,f)-1$ for a sequence $f$).
Then the variance of the autocorrelation demerit factor was studied by Borwein and Lockhart \cite{Borwein-Lockhart}; Aupetit, Liardet, and Slimane \cite{Aupetit-Liardet-Slimane}; Jedwab \cite{Jedwab}; and Katz and Ramirez \cite{Katz-Ramirez}.
The precise variance for the autocorrelation demerit factor was determined by Jedwab, and some of the higher central moments were determined by Katz and Ramirez.

In this paper we give explicit formulae (as a function of sequence length $\ell$) for the variance and the third through sixth central moments of $\SSCC(f,g)$ as $(f,g)$ ranges over $\Pair(\ell)$.
These formulae arise from a general formula for the $p$th central moment of $\SSCC(f,g)$ for every positive integer $p$.
The derivation of this formula requires some delicate combinatorial arguments using group actions on partitions.
The calculation of an explicit functional form in terms of the length $\ell$ becomes more difficult as $p$ increases.
Our theory enables the variance and third central moment to be calculated by hand, while for the fourth, fifth, and sixth central moments, a computer program has automated the calculations.

We first present the formula for the variance of the crosscorrelation demerit factor, which is $\ell^{-4}$ times the variance of $\SSCC$ by \eqref{Celeste}.
The variance of $\SSCC$ is a polynomial function of the sequence length $\ell$, so that the variance of $\CDF$ is a rational function of $\ell$.
\begin{theorem}\label{Valerie}
If $\ell$ is a positive integer, then
\[
\cmomv{2} \CDF(f,g) = \frac{4\ell^3  - 6\ell^2 + 2\ell}{3\ell^4}.
\]
\end{theorem}
Note that in the limit as $\ell$ tends to infinity, the variance of $\CDF$ tends to zero.
When we combine this fact with Sarwate's result (\cref{Abigail}), we see that for large $\ell$, the vast majority of sequences have crosscorrelation demerit factor close to $1$.

The third central moment of the crosscorrelation demerit factor is also a rational function of the sequence length $\ell$.
By \eqref{Celeste}, the third central moment of $\CDF$ is $\ell^{-6}$ times the third central moment of $\SSCC$, which is a polynomial function of $\ell$.
\begin{theorem}\label{Scott}
If $\ell$ is a positive integer, then
\[
\cmomv{3} \CDF(f,g) = \frac{ 8\ell^4 - 32\ell^3 + 40\ell^2 - 16\ell }{\ell^6}.
\]
\end{theorem}

The fourth through sixth central moments of $\SSCC$ are quasipolynomial functions of the sequence length $\ell$.
The fourth central moment of $\SSCC$ is a quasipolynomial of period $2$, and in view of \eqref{Celeste}, one obtains the fourth central moment of $\CDF$ by multiplying this by $\ell^{-8}$.

\begin{theorem}\label{Kurt}
If $\ell$ is a positive integer, then
\begin{align*}
\cmomv{4} \CDF & (f,g) \\
& =\begin{cases}
\frac{ 80\ell^6 + 2208\ell^5 - 17080\ell^4 + 43640\ell^3 - 47020\ell^2 + 18352\ell }{15\ell^8} & \text{if $\ell$ is even,} \\
\frac{ 80\ell^6 + 2208\ell^5 - 17080\ell^4 + 43640\ell^3 - 47380\ell^2 + 19792\ell -1260 }{15\ell^8} & \text{if $\ell$ is odd.}
\end{cases}
\end{align*}
\end{theorem}
The fifth central moment of $\SSCC$ is a quasipolynomial of period $6$, and in view of \eqref{Celeste}, one obtains the fifth central moment of $\CDF$ by multiplying this by $\ell^{-10}$.
\begin{theorem}
If $\ell$ is a positive integer, then
\[
\cmomv{5} \CDF(f,g) = \frac{1}{27\ell^{10}}\sum^7_{j=0}a_j(\ell)\ell^j,
\]
where for every $\ell$ we have $a_7(\ell)=2880$, $a_6(\ell)=140832$, $a_5(\ell)=-1595952$, $a_4(\ell)=6124960$, 
\begin{align*}
a_3(\ell) &= 
\begin{cases}
-9962400 & \text{if $\ell \equiv 0 \pmod{2}$} \\
-9988320 & \text{if $\ell \equiv 1 \pmod{2}$,}
\end{cases}\\
a_2(\ell) &= 
\begin{cases}
6049728 & \text{if $\ell \equiv 0 \pmod{6}$} \\
6195648 & \text{if $\ell \equiv \pm 1 \pmod{6}$} \\
5988288 & \text{if $\ell \equiv \pm 2 \pmod{6}$} \\
6257088 & \text{if $\ell \equiv 3 \pmod{6}$,} 
\end{cases}\\
a_1(\ell) &= 
\begin{cases}
 -753408 & \text{if $\ell \equiv 0 \pmod{6}$} \\
-1001008 & \text{if $\ell \equiv 1 \pmod{6}$} \\
 -435968 & \text{if $\ell \equiv 2 \pmod{6}$} \\
-1420848 & \text{if $\ell \equiv 3 \pmod{6}$} \\
 -333568 & \text{if $\ell \equiv 4 \pmod{6}$} \\
-1103408 & \text{if $\ell \equiv 5 \pmod{6}$, and}
\end{cases}\\
a_0(\ell) &= 
\begin{cases}
      0 & \text{if $\ell \equiv 0 \pmod{6}$} \\
 120960 & \text{if $\ell \equiv 1 \pmod{6}$} \\
 307200 & \text{if $\ell \equiv 2 \pmod{6}$} \\
 673920 & \text{if $\ell \equiv 3 \pmod{6}$} \\
-552960 & \text{if $\ell \equiv 4 \pmod{6}$} \\
 981120 & \text{if $\ell \equiv 5 \pmod{6}$.}
\end{cases}
\end{align*}
\end{theorem}

The sixth central moment of $\SSCC$ is a quasipolynomial of period $60$, and in view of \eqref{Celeste}, one obtains the sixth central moment of $\CDF$ by multiplying this by $\ell^{-12}$.
\begin{theorem}\label{Methuselah}
If $\ell$ is a positive integer, then
\[
\cmomv{6} \CDF(f,g) = \frac{1}{315\ell^{12}}\sum^9_{j=0}a_j(\ell)\ell^j,
\]
where for every $\ell$ we have $a_9(\ell)=11200$, $a_8(\ell)=1179360$, $a_7(\ell)=87442080$, $a_6(\ell)=-1194335520$, 
\begin{align*}
a_5(\ell) &= 
\begin{cases}
3724763000 & \text{if $\ell \equiv 0 \pmod{2}$} \\
3724611800 & \text{if $\ell \equiv 1 \pmod{2}$,}
\end{cases}\\
a_4(\ell) &=
\begin{cases}
13632066280 & \text{if $\ell \equiv 0 \pmod{2}$} \\
13628210680 & \text{if $\ell \equiv 1 \pmod{2}$,}
\end{cases}\\
a_3(\ell) &= 
\begin{cases}
-110959180584 & \text{if $\ell \equiv 0 \pmod{12}$} \\
-111308645784 & \text{if $\ell \equiv \pm 1$ or $\pm 5 \pmod{12}$} \\
-111061962984 & \text{if $\ell \equiv \pm 2 \pmod{12}$} \\
-111214028184 & \text{if $\ell \equiv \pm 3 \pmod{12}$} \\
-111053798184 & \text{if $\ell \equiv \pm 4 \pmod{12}$} \\
-110967345384 & \text{if $\ell \equiv 6 \pmod{12}$,}
\end{cases}\\
a_2(\ell) &=
\begin{cases}
239977889760 & \text{if $\ell \equiv  0 \pmod{12}$} \\
245802750760 & \text{if $\ell \equiv  1 \pmod{12}$} \\
242271918560 & \text{if $\ell \equiv  2 \pmod{12}$} \\
243687619560 & \text{if $\ell \equiv  3 \pmod{12}$} \\
242090299360 & \text{if $\ell \equiv  4 \pmod{12}$} \\
245672293160 & \text{if $\ell \equiv  5 \pmod{12}$} \\
240289966560 & \text{if $\ell \equiv  6 \pmod{12}$} \\
245800029160 & \text{if $\ell \equiv  7 \pmod{12}$} \\
241959841760 & \text{if $\ell \equiv  8 \pmod{12}$} \\
243690341160 & \text{if $\ell \equiv  9 \pmod{12}$} \\
242402376160 & \text{if $\ell \equiv 10 \pmod{12}$} \\
245669571560 & \text{if $\ell \equiv 11 \pmod{12}$,}
\end{cases}
\end{align*}
and $a_1(\ell)$ and $a_0(\ell)$ are functions of period $60$ whose values are given on Tables \ref{Tiffany} and \ref{Tony}, respectively.
\end{theorem}
\begin{table}[ht!]
\begin{center}
\rowcolors{2}{gray!20}{white}
\caption{Values of $a_1(\ell)$ from \cref{Methuselah} as a function of $\ell \pmod{60}$}\label{Tiffany}
\begin{tabular}{|l|r||l|r||l|r|}
\hline
$\ell \pmod{60}$ & $a_1(\ell)$ & $\ell \pmod{60}$ & $a_1(\ell)$ & $\ell \pmod{60}$ & $a_1(\ell)$ \\
\hline
$ 0$       & $-161841373440$ & $12$, $48$ & $-163179094272$ & $28$, $52$ & $-176989679872$ \\
$ 1$, $49$ & $-191149692928$ & $13$, $37$ & $-191595599872$ & $29$, $41$ & $-187608700928$ \\
$ 2$, $38$ & $-178672345472$ & $14$, $26$ & $-178226438528$ & $30$       & $-167065031040$ \\
$ 3$, $27$ & $-177204406272$ & $15$       & $-175866685440$ & $33$, $57$ & $-177785014272$ \\
$ 4$, $16$ & $-176543772928$ & $17$, $53$ & $-188054607872$ & $34$, $46$ & $-181767430528$ \\
$ 5$       & $-186716887040$ & $18$, $42$ & $-168402751872$ & $35$       & $-186136279040$ \\
$ 6$, $54$ & $-167956844928$ & $19$, $31$ & $-190569084928$ & $39$, $51$ & $-176758499328$ \\
$ 7$, $43$ & $-191014991872$ & $20$       & $-172110967040$ & $40$       & $-175651959040$ \\
$ 8$, $32$ & $-173448687872$ & $22$, $58$ & $-182213337472$ & $44$, $56$ & $-173002780928$ \\
$ 9$, $21$ & $-177339107328$ & $23$, $47$ & $-187473999872$ & $45$       & $-176447293440$ \\
$10$       & $-180875616640$ & $24$, $36$ & $-162733187328$ & $50$       & $-177334624640$ \\
$11$, $59$ & $-187028092928$ & $25$       & $-190257879040$ & $55$       & $-189677271040$ \\
\hline
\end{tabular}
\end{center}
\end{table}

\begin{table}[ht!]
\begin{center}
\rowcolors{2}{gray!20}{white}
\caption{Values of $a_0(\ell)$ from \cref{Methuselah} as a function of $\ell \pmod{60}$}\label{Tony}
\begin{tabular}{|l|r||l|r||l|r|}
\hline
$\ell \pmod{60}$ & $a_0(\ell)$ & $\ell \pmod{60}$ & $a_0(\ell)$ & $\ell \pmod{60}$ & $a_0(\ell)$ \\
\hline
$ 0$ & $          0$ & $20$ & $-3579699200$ & $40$ & $21342003200$ \\
$ 1$ & $40408468352$ & $21$ & $19066465152$ & $41$ & $15486765952$ \\
$ 2$ & $ 4384500736$ & $22$ & $29306203136$ & $42$ & $ 7964199936$ \\
$ 3$ & $ 3755433024$ & $23$ & $  175733824$ & $43$ & $25097436224$ \\
$ 4$ & $20190076928$ & $24$ & $-1151926272$ & $44$ & $-4731625472$ \\
$ 5$ & $12848483200$ & $25$ & $37770185600$ & $45$ & $16428182400$ \\
$ 6$ & $ 8168573952$ & $26$ & $ 4588874752$ & $46$ & $29510577152$ \\
$ 7$ & $26992540736$ & $27$ & $ 5650537536$ & $47$ & $ 2070838336$ \\
$ 8$ & $-3040894976$ & $28$ & $21880807424$ & $48$ & $  538804224$ \\
$ 9$ & $15276256128$ & $29$ & $11696556928$ & $49$ & $36618259328$ \\
$10$ & $26872294400$ & $30$ & $ 5530291200$ & $50$ & $ 1950592000$ \\
$11$ & $ 2275212352$ & $31$ & $27196914752$ & $51$ & $ 5854911552$ \\
$12$ & $ 2433908736$ & $32$ & $-1145790464$ & $52$ & $23775911936$ \\
$13$ & $38308989824$ & $33$ & $16966986624$ & $53$ & $13387287424$ \\
$14$ & $  798665728$ & $34$ & $25720368128$ & $54$ & $ 4378364928$ \\
$15$ & $ 3216628800$ & $35$ & $ -363070400$ & $55$ & $24558632000$ \\
$16$ & $23980285952$ & $36$ & $ 2638282752$ & $56$ & $ -941416448$ \\
$17$ & $15282391936$ & $37$ & $40204094336$ & $57$ & $18862091136$ \\
$18$ & $ 6069095424$ & $38$ & $ 2489396224$ & $58$ & $27411098624$ \\
$19$ & $23406705728$ & $39$ & $ 2064702528$ & $59$ & $-1514996672$ \\
\hline
\end{tabular}
\end{center}
\end{table}

Another fact that we learn from our general formula for the central moments of $\SSCC$ is that the $p$th central moments are always nonnegative, and (except for the first central moment, which is always zero) they become positive for sufficiently long sequences.  This translates into the following result about the moments of $\CDF$.
\begin{theorem}\label{Priscilla}
If $\ell$ and $p$ are positive integers, then $\mu^\ell_{p,(f,g)} \CDF(f,g)$ is nonnegative.  Moreover, if (i) $p=1$,  (ii) if $p$ is odd and $\ell \leq 2$, or (iii) if $p$ is even and $\ell \leq 1$, then $\mu^\ell_{p,(f,g)} \CDF(f,g)$ equals $0$; otherwise it is strictly positive.
\end{theorem}
In particular, for $p=3$, this result tells us that the distribution is always right-skewed as long as $\ell > 2$.

The rest of this paper is organized as follows.
\cref{Patricia} establishes the notations, definitions, and basic results that give us the language we need to express our results on central moments of $\SSCC$.
Most of this section concerns certain functions called {\it assignments} as well as partitions of the domains of these assignments.
\cref{Fred} derives an exact formula (in \cref{Francis}) for the $p$th central moment of $\SSCC$ in terms of these assignments and partitions.
The number of terms in our formula grows very rapidly as $p$ increases, but many terms give the same contribution due to symmetries of the partitions that underly the terms.
This motivates us to introduce the notion of coherently isomorphic partitions, which is developed in \cref{Isabel} using a particular group action.
\cref{Isabel} concludes with another exact formula (in \cref{Vito}) for the $p$th central moment of $\SSCC$, but this formula groups terms according to coherent isomorphism classes of the underlying partitions.
This formula makes computations easier, but in general it is still a nontrivial task to find all the coherent isomorphism classes.
\cref{Theresa} provides a matrix-based method, which can be implemented as a computer algorithm, for finding these classes.
\cref{Portia} gives a proof of our positivity result (\cref{Priscilla}).
Sections \ref{Anne} and \ref{Sidney} respectively provide hand calculations of the second and third central moments of $\SSCC$; these yield the variance of $\CDF$ in \cref{Valerie} and the third central moment of $\CDF$ in \cref{Scott}, from which we can compute the skewness of $\CDF$.
\cref{Clarence} reports on computer-assisted calculations of the fourth through sixth central moments of $\SSCC$, which yield the fourth through sixth central moments of $\CDF$ as reported above in Theorems \ref{Kurt}--\ref{Methuselah}.
There is also an appendix with some algebraic lemmas that are used in the proof of \cref{Francis}.

\section{Preliminaries}\label{Patricia}

In this paper, $\N$ always denotes the set $\{0,1,\ldots\}$ of nonnegative integers, while $\Z_+$ denotes the set $\{1,2,\ldots\}$ of strictly positive integers.  If $\ell \in \N$, then $[\ell]$ denotes the set $\{0,1,\ldots,\ell-1\}$.
If $A$ and $B$ are sets, then $B^A$ denotes the set of all functions from $A$ into $B$.
If $k \in \Z_+$ and we are writing $k$-tuples of integers, we sometimes omit the enclosing parentheses and sometimes also omit the commas if there is no risk of ambiguity.
For example, if we know that the context demands a triple of integers, then writing $010$ is shorthand for $(0,1,0)$, but we could not do this if any of our integers required more than one digit in its representation.
Similarly, if the context demands a triple and we are representing the terms as variables, we could write $x y z$ as a shorthand for $(x,y,z)$.

For $p \in \N$, much of our discussion involves triples $(e,s,v) \in \indexset$ because calculation of the $p$th central moment involves systems of $p$ equations with two terms on each side.
Then $(e,s,v)$ indexes the $v$th term (0 for the left-hand term, 1 for the right-hand term) on the $s$th side ($0$ indicates left-hand side, $1$ indicates right-hand side) of equation number $e$.
We employ the following useful definition from \cite[Def.~2.1]{Katz-Ramirez}, which concerns functions whose inputs are these triples.
\begin{definition}[Assignment]\label{Arnold}
Let $E\subseteq\N$.  An {\it assignment for $E$} is a function from $\eindexset$ into $\N$, i.e., an element of $\N^\eindexset$.
If $\tau$ is an assignment for $E$ and $(e,s,v) \in \eindexset$, then we use subscript notation $\tau_{e,s,v}$ (or just $\tau_{e s v}$) in place of the more usual $\tau(e,s,v)$ to denote the value of $\tau$ at $(e,s,v)$.  We use the following notations for the set of all assignments for $E$ and for useful subsets thereof:
\begin{itemize}
\item $\As(E)=\N^\eindexset$,
\item $\As(E,=)=\{\tau \in \As(E): \tau_{e 0 0}+\tau_{e 0 1}=\tau_{e 1 0}+\tau_{e 1 1} \text{ for all $e \in E$}\}$,
\item $\As(E,\ell)=\{\tau \in \As(E): \tau(\eindexset) \subseteq [\ell]\}$, and
\item $\As(E,=,\ell)=\As(E,=)\cap\As(E,\ell)$.
\end{itemize}
\end{definition}
See \cite[Ex.~2.2]{Katz-Ramirez} for an example.
As we shall see in \cref{Fred}, the determination of central moments of $\SSCC$ often requires us to count assignments, and it helps to group assignments according to properties of their fibers.
The nonempty fibers of a function partition the function's domain, and so we introduce some concepts and notations for partitions.
If $\cP$ is a partition of a set $A$ and $a,b \in A$, then we write $a \equiv b \pmod{\cP}$ to mean that $a$ and $b$ are elements of the same class in $\cP$ (i.e., there is some $P \in \cP$ with $a,b \in P$).
This means that if $f$ is a function whose domain is $A$ and $\cP$ is the partition whose classes are the nonempty fibers of $f$, then $a \equiv b \pmod{\cP}$ if and only if $f(a)=f(b)$.
We also introduce some notations for partitions of domains of assignments.
\begin{definition}[$\Part(E)$ and $\Part(p)$] Let $E\subseteq \N$.  Then $\Part(E)$ is the set of partitions of $\eindexset$.  If $p \in \N$, then $\Part(p)$ is a shorthand for $\Part([p])$.
\end{definition}
In the rest of this section we relate certain partitions to assignments and investigate properties of these partitions that are critical in the combinatorial determination of the central moments of $\SSCC$.
The first of these properties is called coherence.
\begin{definition}[Coherent set, coherent partition, $\Coh(E)$, and $\Coh(p)$]  Let $E\subseteq\N$.  If $v \in [2]$, then we say that a subset $P$ of $E\times\bindexset$ is {\it $v$-coherent} to mean that $P$ is a subset of $E\times[2]\times\{v\}$.
A {\it coherent} subset of $E\times\bindexset$ is any subset that is $v$-coherent for some $v \in [2]$.
If $\cP\in\Part(E)$, then we say that $\cP$ is {\it coherent} to mean that all the classes in $\cP$ are coherent.
We write $\Coh(E)$ for the set of all coherent partitions in $\Part(E)$.
If $p \in \N$, then $\Coh(p)$ is a shorthand for $\Coh([p])$.
\end{definition}
\begin{example}\label{Electra}
Note that $\{(0,0,1),(1,0,1)\}$ is a coherent subset (in fact, a $1$-coherent subset) of $[2]\times\bindexset$, while $\{(0,0,0),(0,0,1),(1,0,0),(1,0,1)\}$ is not coherent because it has triples whose third coordinates differ.
Therefore,
\[
\left\{\{(0,0,0),(0,0,1),(1,0,0),(1,0,1)\},\{(0,1,0),(1,1,0)\},\{(0,1,1),(1,1,1)\}\right\} \in \Part(2)
\]
is not a coherent partition (i.e., is not in $\Coh(2)$).
But
\[
\cP=\left\{\{(0,0,0),(1,0,0))\},\{(0,0,1),(1,0,1)\},\{(0,1,0),(1,1,0)\},\{(0,1,1),(1,1,1)\}\right\} \in \Coh(2)
\]
because every class in this partition is coherent.
\end{example}
Now we can describe the partition that we associate to a given assignment.
\begin{definition}[Partition coherently induced by an assignment]
Let $E\subseteq\N$ and $\tau\in \As(E)$.
Then the {\it partition coherently induced by $\tau$} is the partition
\[
\left\{\tau^{-1}(\{j\})\cap(E\times[2]\times\{v\}): j \in \N, v \in [2]\right\} \smallsetminus\{\emptyset\},
\]
which is $\Coh(E)$.
\end{definition}
\begin{example}\label{Agamemnon}
Let $\tau \in \As([2])$ be such that $\tau_{000}=\tau_{100}=\tau_{011}=\tau_{111}=5$ and $\tau_{001}=\tau_{101}=\tau_{010}=\tau_{110}=6$.
Then the partition coherently induced by $\tau$ is in fact the partition $\cP$ from \cref{Electra}.  This is because the nonempty fibers of $\tau$ are $\{(0,0,0),(1,0,0)),(0,1,1),(1,1,1)\}$ and $\{(0,0,1),(1,0,1),(0,1,0),(1,1,0)\}$ and when we intersect these with $[2]\times[2]\times\{0\}$ and with $[2]\times[2]\times\{1\}$, we get the four classes of $\cP$.
\end{example}
Now we build on the notations of \cref{Arnold} to give sets of assignments that coherently induce a particular coherent partition.
\begin{definition}[Assignments coherently inducing a coherent partition]\label{Colette}
Let $E\subseteq\N$ and let $\cP\in\Coh(E)$.  Then
\begin{itemize}
\item $\CAs(\cP)$ is the set of all $\tau \in \As(E)$ that coherently induce $\cP$, that is, the set of all $\tau\in\As(E)$ such that for every $(e,s,v),(e',s',v')\in\eindexset$ with $v=v'$, we have $\tau_{e,s,v}=\tau_{e',s',v'}$ if and only if $(e,s,v)\equiv(e',s',v') \pmod{\cP}$.
\item $\CAs(\cP,=)=\CAs(\cP)\cap\As(E,=)$;
\item $\CAs(\cP,\ell)=\CAs(\cP)\cap\As(E,\ell)$; and
\item $\CAs(\cP,=,\ell)=\CAs(\cP)\cap\As(E,=)\cap\As(E,\ell)$.
\end{itemize}
\end{definition}
The following technical result sometimes guarantees the nonexistence of certain assignments.
\begin{lemma}\label{Wayne}
Let $p\in\N$ and let $\cP\in\Coh(p)$, and for each $v \in [2]$, let $N_v$ be the number of classes in $\cP$ that are subsets of $[p]\times[2]\times\{v\}$.
Then $\CAs(\cP,\ell)=\emptyset$ if $\ell < \max\{N_0,N_1\}$.
\end{lemma}
\begin{proof}
Assume that $\ell < \max\{N_0,N_1\}$ and let $v \in [2]$ be such that $\ell < N_v$.
Label the classes in $\cP$ that are subsets of $[p]\times[2]\times\{v\}$ as $P_1,\ldots,P_{N_v}$ and select some $\alpha_j \in P_j$ for each $j \in \{1,\ldots,N_v\}$.
Suppose that there is some $\tau \in \CAs(\cP,\ell)$ to show contradiction.
Then $\tau_{\alpha_1},\tau_{\alpha_2},\ldots,\tau_{\alpha_{N_v}}$ must be distinct values in $[\ell]$, but $[\ell]$ has fewer than $N_v$ elements in it, so we must have $\CAs(\cP,\ell)=\emptyset$.
\end{proof}
The second property critical to our combinatorial determination of moments is coherent satisfiability.
\begin{definition}[Coherently satisfiable partition]
Let $E\subseteq\N$ and let $\cP\in\Coh(E)$.
Then we say that $\cP$ is {\it coherently satisfiable} to mean that $\CAs(\cP,=)$ is nonempty, or equivalently, to mean that there exists $\ell \in \N$ such that $\CAs(\cP,=,\ell)$ is nonempty. Let $\CSat(E)$ denote the set of coherently satisfiable partitions that are elements of $\Part(E)$, and if $p \in \N$, then $\CSat(p)$ is a shorthand for $\CSat([p])$.
\end{definition}
\begin{example}\label{Achilles}
Consider the partition $\cP$ from \cref{Electra}, and recall that the assignment $\tau$ from \cref{Agamemnon} coherently induces $\cP$.  Thus, $\tau\in\CAs(\cP)$.
Note that $\tau_{e00}+\tau_{e01}=\tau_{e10}+\tau_{e11}$ for each $e \in [2]$ for this assignment, so that $\tau\in\CAs(\cP,=)$ (and indeed is in $\CAs(\cP,=,\ell)$ for all $\ell > 6$ since the image of $\tau$ is $\{5,6\}$).  Thus, $\cP$ is coherently satisfiable, and so $\cP \in \CSat(2)$.
\end{example}
To prepare ourselves to discuss the third property of partitions that is critical to the determination of our moments, we now need to introduce the notion of restriction of assignments and associated notations for the sets and partitions related to these assignments.
\begin{definition}[Restriction of functions, sets, and partitions]
Let $F \subseteq E \subseteq \N$.
If $\tau\in\As(E)$, then the {\it restriction of $\tau$ to $F$}, denoted $\tau\vert_F$, indicates the function from $F \times[2]\times [2]$ to $\N$ with $\tau\vert_F (e,s,v)=\tau(e,s,v)$ for every $(e,s,v) \in F\times[2]\times [2]$; that is, $\tau\vert_F$ is the function obtained from $\tau$ by restricting the domain to $F\times[2]\times[2]$.
If $P$ is a subset of $\eindexset$, then the {\it restriction of $P$ to $F$}, written $P_F$, is $P\cap(F\times[2]\times[2])$.
If $\cP$ is a partition of $\eindexset$, then the {\it restriction of $\cP$ to $F$}, written $\cP_F$, is the partition $\{P_F : P \in \cP\}\smallsetminus\{\emptyset\}$ of $F\times [2]\times [2]$.
If $\fP$ is a set of partitions of $\eindexset$, then {\it restriction of $\fP$ to $F$}, written $\fP_F$, is the set of partitions $\{\cP_F: \cP \in \fP\}$.
\end{definition}
See \cite[Ex.~2.11]{Katz-Ramirez} for examples.
We now see that restriction respects coherence of partitions.
\begin{lemma}
Let $F\subseteq E\subseteq\N$ and let $\cP\in\Coh(E)$.  Then $\cP_F \in \Coh(F)$.
\end{lemma}
\begin{proof}
By the way we define restriction, $\cP_F\in\Part(F)$.  Let $(f,s,v),(f',s',v') \in F\times[2]\times[2]$ with $v\not=v'$.  Then the coherence of $\cP$ means that $(f,s,v)$ and $(f',s',v')$ reside in different classes of $\cP$, and so they also reside in different classes of $\cP_F$.
\end{proof}
Restriction also respects the various properties of assignments from Definitions \ref{Arnold} and \ref{Colette}.
\begin{lemma}\label{Daphne}
Let $\ell\in\N$ and $F\subseteq E\subseteq \N$, and $\cP\in\Coh(E)$.
\begin{enumerate}[label=(\roman*)]
\item\label{James} If $\tau \in \As(E)$, then $\tau\vert_F \in \As(F)$.
\item\label{Katherine} If $\tau \in \As(E,\ell)$, then $\tau\vert_F \in \As(F,\ell)$.
\item\label{Larry} If $\tau \in \As(E,=)$, then $\tau\vert_F \in \As(F,=)$.
\item\label{Mary} If $\tau \in \As(E,=,\ell)$, then $\tau\vert_F \in \As(F,=,\ell)$.
\item\label{Nancy} If $\tau \in \CAs(\cP)$, then $\tau\vert_F \in \CAs(\cP_F)$.
\item\label{Orestes} If $\tau \in \CAs(\cP,\ell)$, then $\tau\vert_F \in \CAs(\cP_F,\ell)$.
\item\label{Peter} If $\tau \in \CAs(\cP,=)$, then $\tau\vert_F \in \CAs(\cP_F,=)$.
\item\label{Quentin} If $\tau \in \CAs(\cP,=,\ell)$, then $\tau\vert_F \in \CAs(\cP_F,=,\ell)$.
\end{enumerate}
\end{lemma}
\begin{proof}
Parts \ref{James}--\ref{Mary} are proved in \cite[Lemma ~2.12]{Katz-Ramirez}.

To prove \ref{Nancy}, let $\tau\in\CAs(\cP)$, so that for every $(e,s,v),(e',s',v') \in \eindexset$ with $v=v'$ we have $\tau_{e,s,v} = \tau_{e',s',v'}$ if and only if $(e,s,v)\equiv(e',s',v') \pmod{\cP}$. Thus, when $(e,s,v),(e',s',v')\in F\times [2]\times[2]$ with $v=v'$, we have $(\tau\vert_F)_{e,s,v} = (\tau\vert_F)_{e',s',v'}$ if and only if $(e,s,v)\equiv(e',s',v') \pmod{\cP_F}$. Hence, we have $\tau\vert_F\in \CAs(\cP_F)$.
To prove \ref{Orestes}, \ref{Peter}, and \ref{Quentin} respectively, then use \ref{Katherine}, \ref{Larry}, and \ref{Mary} each together with \ref{Nancy}.
\end{proof}
Partitions in which all classes are of even cardinality play an important role in our calculations.
\begin{definition}[Even partition]
We say that a partition $\cP$ of some set $S$ is {\it even} to mean that every class $P \in \cP$ is a finite set of even cardinality.
\end{definition}
See \cite[Ex.~2.15]{Katz-Ramirez} for an example.
Now we define the third important property that the key partitions in our calculation of central moments have.
\begin{definition}[Globally even, locally odd partition]
Let $E\subseteq \N$.
We say that a partition $\cP$ of $\eindexset$ is {\it globally even, locally odd} (abbreviated {\it GELO}) to mean that $\cP$ is an even partition such that for every $e \in E$, the restricted partition $\cP_{\{e\}}$ is not even.
For $p \in \N$, we use $\GELO(p)$ to denote the set of GELO partitions in $\Part(p)$.
\end{definition}
See \cite[Ex.~2.17]{Katz-Ramirez} for an example.
We pause to record the fact that when $p=1$, the notions of global and local coincide, so the GELO property cannot be satisfied.
\begin{lemma}\label{Walter}
We have $\GELO(1)=\emptyset$.
\end{lemma}
\begin{proof}
Any partition $\cP\in\Part(1)$ has $\cP_{\{0\}}=\cP$, so we cannot have $\cP$ even and $\cP_{\{0\}}$ non-even, so $\GELO(1)=\emptyset$.
\end{proof}
We now define contributory coherence, which is the synthesis of our three important properties and which is the key property of partitions in our calculation of central moments.
\begin{definition}[Coherently contributory partition]
Let $p \in \N$ and let $\cP\in\Coh(p)$.
We say that $\cP$ is {\it coherently contributory} to mean that $\cP$ is both coherently satisfiable and globally even, locally odd.
Let $\CCon(p)$ denote the set of contributory elements of $\Coh(p)$, so that $\CCon(p)=\CSat(p)\cap\GELO(p)$.
\end{definition}
\begin{example}\label{Ajax}
Consider the partition $\cP$ from \cref{Electra}.
Note that that $\cP$ is an even partition (every class has cardinality $2$) but that $\cP_{\{0\}}=\left\{\{(0,0,0)\},\{(0,0,1)\},\{(0,1,0)\},\{(0,1,1)\}\right\}$ and $\cP_{\{1\}}=\left\{\{(1,0,0)\},\{(1,0,1)\},\{(1,1,0)\},\{(1,1,1)\}\right\}$, neither of which is an even partition.
So $\cP$ is GELO, and since \cref{Achilles} shows that $\cP \in \CSat(2)$, we conclude that $\cP \in \CCon(2)$.
\end{example}
We close this section with a result about structural features of coherently contributory partitions.
The following terminology makes these features easier to discuss.
\begin{definition}[Atomized set]
Let $p \in \N$.  An {\it atomized set} is any subset $P$ of $\indexset$ such that $|P_{\{e\}}|\leq 1$ for every $e \in [p]$.
\end{definition}
Now we indicate some constraints on the structure of coherently contributory partitions.
\begin{lemma}\label{Alberto}
Suppose that $p\in\N$ and $\cP\in\CCon(p)$.
Then the following hold.
\begin{enumerate}[label=(\roman*)]
\item\label{Bertram} For each $e \in [p]$, we have
\[
\cP_{\{e\}}=\left\{\{(e,0,0)\},\{(e,0,1)\},\{(e,1,0)\},\{(e,1,1)\}\right\},
\]
i.e., every class in $\cP$ is atomized.
\item\label{Corinna} No class in $\cP$ has more than $p$ elements.
\item\label{David} If $p>0$, then $\card{\cP}$ is at least $4$ but not greater than $2 p$.  Furthermore, if $p$ is odd, then $\card{\cP}$ is at least $6$.
\end{enumerate}
\end{lemma}
\begin{proof}
Since $\cP\in\CCon(p)\subseteq\GELO(p)\cap\Coh(p)$, for every $e\in [p]$, the restriction $\cP_{\{e\}}$ is a noneven partition of $\{(e,0,0),(e,0,1),(e,1,0),(e,1,1)\}$ into classes where no class can contain elements with differing third coordinates.
So if the partition $\cP_{\{e\}}$ has a class with more than one element, $\cP_{\{e\}}$ must be of the form $\{\{(e,0,v),(e,1,v)\},\{(e,0,1-v)\},\{(e,1,1-v)\}\}$ for some $v \in [2]$.
But then since $\cP\in\CCon(p)\subseteq\CSat(p)$, there is some $\tau\in\CAs(\cP,=)$, so that $\tau_{e,0,v}=\tau_{e,1,v}$ and $\tau_{e,0,1-v}\not=\tau_{e,1,1-v}$ and yet $\tau_{e,0,v}+\tau_{e,0,1-v}=\tau_{e,1,v}+\tau_{e,1,1-v}$, so we obtain a contradiction.
So we conclude that all classes of $\cP_{\{e\}}$ are singleton classes, so that $\cP_{\{e\}}$ has the claimed form.
It then follows that no class of $\cP$ can have more than $p$ elements because for each $e \in [p]$, the class has at most one element in $\{e\}\times[2]\times[2]$.

From now on, suppose that $p>0$.
Since $0 \in [p]$ and $\cP_{\{0\}}$ cannot have more classes than $\cP$, we have $\card{\cP}\geq 4$.
On the other hand, since $\cP$ is a partition of the $4 p$ elements of $\indexset$ into nonempty classes of even cardinality, we see that $\card{\cP}\leq 4 p/2=2 p$.

Now suppose that $p$ is odd, and we may assume $p>1$ since $\CCon(1)\subseteq\GELO(1)=\emptyset$ by \cref{Walter}.
From part \ref{Corinna} we know that no class of $\cP$ has more than $p$ elements, and since $p$ is odd but $\cP$ is even (since $\cP\in\CCon(p)\subseteq\GELO(p)$), this means that no class $\cP$ has more than $p-1$ elements.
Then $\cP\in\CCon(p)\subseteq\Coh(p)$, so it is a union of a partition of $[p]\times[2]\times\{0\}$ and a partition of $[p]\times[2]\times\{1\}$.
So each of these two partitions has at least $\ceil{2 p/(p-1)}=3$ classes, and so $\cP$ has at least $6$ classes.
\end{proof}

\section{Centralized moment formula}\label{Fred}

This section details the theory for calculating a formula for the $p$th central moment of $\SSCC(f,g)$, the sum of squares of the crosscorrelation values for a pair of binary sequences $f$ and $g$.
Recall from the Introduction that the moments are computed with $(f,g)$ ranging over the set of all pairs of length $\ell$ binary sequences, $\Pair(\ell)$, with uniform probability distribution. Also, recall our convention that $\ev v(f,g)={\mathbf E}_{(f,g) \in \Pair(\ell)}(v(f))$ denotes the expected value of a random variable of this distribution, and the $p$th central moment of this random variable is 
\[
\cmom v(f,g) = \ev\left(v(f,g) - \ev v(f,g)\right)^p.
\]

We are working with the $\SSCC$ because it is mathematically more convenient than $\CDF$, and since \eqref{Celeste} shows how to convert between $\CDF$ and $\SSCC$, it is simple to convert the results on $\SSCC$ here to results in terms of $\CDF$.
\cref{Francis} details a method for calculating the central moments of $\SSCC$ in terms of coherently contributory partitions and assignments.
The proof relies on several lemmas, which are included and proved in \cref{Cat}.

\begin{proposition}\label{Francis}
For $p,\ell \in \N$, we have
\[
\cmom\SSCC(f,g) = \sum_{\cP \in \CCon(p)} \card{\CAs(\cP,=,\ell)}.
\]
\end{proposition}
\begin{proof}
Let $f=(\ldots,f_0,f_1,f_2,\ldots)$ and $g=(\ldots,g_0,g_1,g_2,\ldots)$ be binary sequences of length $\ell$, so that $f_j=g_j=0$ when $j\not\in[\ell]$, and let $h$ be the pair of sequences $(f,g)$.
If $E\subseteq\N$ and $\tau\in\As(E)$, then we use the notation $h^{\tau}$ as a shorthand for $\left(\prod_{\alpha \in \leindexset} f_{\tau_\alpha}\right)\left(\prod_{\beta \in \reindexset} g_{\tau_\beta}\right)$.
By \cite[eq.~(14)]{Katz}, we have
\[
\SSCC(h) = \sums{s,t,u,v \in \Z \\ s+t=u+v} f_s g_t f_u g_v = \sums{s,t,u,v \in [\ell] \\ s+t=u+v} f_s g_t f_u g_v.
\]
Using our notation, this formula becomes $\SSCC(f,g)=\sum_{\tau \in As(\{e\},=,\ell)} h^\tau$ for every $e \in [p]$.
Therefore
\[
\cmomh\SSCC(h) = \evh\left(\prod_{e \in [p]} \left(\sum_{\tau^{(e)} \in \As(\{e\},=,\ell)} h^{\tau^{(e)}} - \evh\left(\sum_{\tau^{(e)} \in \As(\{e\},=,\ell)} h^{\tau^{(e)}}\right) \right)\right).
\]
Then by a binomial expansion, we obtain
\[
\cmomh\SSCC(h) = \evh\left(\sum_{E \subseteq [p]} \left(\prod_{d \in [p]\smallsetminus E} \,\,\, \sum_{\tau^{(d)} \in \As(\{d\},=,\ell)} h^{\tau^{(d)}}\right) (-1)^{\card{E}}\prod_{e \in E} \left(\evh\sum_{\tau^{(e)} \in \As(\{e\},=,\ell)} h^{\tau^{(e)}}\right)\right).
\]
Then by \cref{Gerald}, we obtain
\begin{align*}
\cmomh\SSCC(h)
& = \evh\left(\sum_{E \subseteq [p]} \left(\sum_{\upsilon \in \As([p]\smallsetminus E,=,\ell)} h^\upsilon\right) (-1)^{\card{E}}\prod_{e \in E} \left(\evh\sum_{\tau^{(e)} \in \As(\{e\},=,\ell)} h^{\tau^{(e)}}\right)\right) \\
& = \sum_{E \subseteq [p]} (-1)^{\card{E}} \left(\sum_{\upsilon \in \As([p]\smallsetminus E,=,\ell)} \evh\left(h^\upsilon\right)\right) \prod_{e \in E} \left(\sum_{\tau^{(e)} \in \As(\{e\},=,\ell)} \evh\left(h^{\tau^{(e)}}\right)\right).
\end{align*}
Then by \cref{Hortense}, we have
\[
\cmomh\SSCC(h) = \sum_{E \subseteq [p]} (-1)^{\card{E}} \sum_{\tau \in \As([p],=,\ell)} \evh\left(h^{\tau\vert_{[p]\smallsetminus E}}\right) \prod_{e \in E} \evh\left(h^{\tau\vert_{\{e\}}}\right).
\]
Now we sort the assignments in $\As([p],=,\ell)$ according to which partition they coherently induce.
\begin{align*}
\cmomh\SSCC(h)
& = \sum_{E \subseteq [p]} (-1)^{\card{E}} \sum_{\cP \in\Coh(p)} \,\,\, \sum_{\tau \in \CAs(\cP,=,\ell)} \evh\left(h^{\tau\vert_{[p]\smallsetminus E}}\right) \prod_{e \in E} \evh\left(h^{\tau\vert_{\{e\}}}\right) \\
& =  \sum_{\cP \in\Coh(p)} \,\,\, \sum_{E \subseteq [p]} (-1)^{\card{E}} \sum_{\tau \in \CAs(\cP,=,\ell)} \evh\left(h^{\tau\vert_{[p]\smallsetminus E}}\right) \prod_{e \in E} \evh\left(h^{\tau\vert_{\{e\}}}\right).
\end{align*}
Let $\cP$ be a fixed partition in $\Coh(p)$, and consider the value of the innermost sum in the last expression.
If $\tau\in\CAs(\cP,=,\ell)$, then \cref{Irene} tells us that
\begin{equation}\label{Eleanor}
\evh\left(h^{\tau\vert_{[p]\smallsetminus E}}\right) \prod_{e \in E} \evh\left(h^{\tau\vert_{\{e\}}}\right)
\end{equation}
equals $1$ if and only if $\cP_{\{e\}}$ is even for every $e \in E$ and also $\cP_{[p]\smallsetminus E}$ is even.
Otherwise the term \eqref{Eleanor} is $0$.
Equivalently, the term \eqref{Eleanor} is $1$ if and only if $\cP_{\{e\}}$ is even for every $e \in E$ and $\cP$ is even.
So we fix an even $\cP \in \Coh(p)$, and let $C$ be the set of all $c \in [p]$ such that $\cP_{\{c\}}$ is even, and then
\begin{align*}
\sum_{E \subseteq [p]} (-1)^{\card{E}} \sum_{\tau \in \CAs(\cP,=,\ell)} \evh\left(h^{\tau\vert_{[p]\smallsetminus E}}\right) \prod_{e \in E} \evh\left(h^{\tau\vert_{\{e\}}}\right)
& = \sum_{E \subseteq C} (-1)^{\card{E}} \card{\CAs(\cP,=,\ell)} \\
& = \card{\CAs(\cP,=,\ell)}  \sum_{E \subseteq C} (-1)^{\card{E}}  \\
& = \begin{cases}
\card{\CAs(\cP,=,\ell)} & \text{if $C=\emptyset$} \\
0 & \text{otherwise.}
\end{cases}
\end{align*}
So we only get a nonzero contribution from partitions that are globally even, locally odd.
So our central moment calculation becomes
\begin{align*}
\cmomh\SSCC(h)
& =  \sum_{\cP \in\Coh(p)\cap\GELO(p)} \card{\CAs(\cP,=,\ell)} \\
& =  \sum_{\cP \in\CCon(p)} \card{\CAs(\cP,=,\ell)},
\end{align*}
where the final step uses the fact that $\CAs(\cP,=,\ell)$ is empty whenever $\cP$ is not coherently satisfiable.
\end{proof}

\section{Coherent isomorphism classes of partitions}\label{Isabel}

The formula for the $p$th central moment of $\SSCC$ in \cref{Francis} can be arduous to apply because the set $\CCon(p)$ contains a vast number of partitions when $p$ becomes large (already more than one thousand partitions when $p=4$).
There are symmetry transformations acting on our partitions $\cP \in \CCon(p)$ that preserve the associated contributions $\card{\CAs(\cP,=,\ell)}$ in the formula for the $p$th central moment in \cref{Francis}.
By introducing the appropriate group of symmetries and organizing $\CCon(p)$ into orbits under the action of this group, we can obtain a summation formula for the $p$th central moment of $\SSCC$ with far fewer terms than the one in \cref{Francis}.
This section is dedicated to defining this group, describing its action on partitions, and using this to derive a new formula for the $p$th central moment of $\SSCC$ that is easier to use than the one in \cref{Francis}.

We begin with group-theoretic notations and conventions.
If $A$ is a set, then $S_A$ denotes the group of all permutations of $A$.
If $p\in\N$, then $S_p$ is shorthand for $S_{[p]}$.
If $S$ is a subset of $\indexset$ and $\pi\in S_{\indexset}$, then $\pi(S)=\{\pi(s): s \in S\}$.
If $\cS$ is a set of subsets of $\indexset$, then $\pi(\cS)$ is a shorthand for $\{\pi(S): S \in \cS\}$.
If $\fS$ is a set of sets of subsets of $\indexset$, then $\pi(\fS)$ is a shorthand for $\{\pi(\cS): \cS \in \fS\}$.

For $p\in\N$, we let $\Wrp$ be the wreath product $S_2 \wr_{[p]} S_p$.
For $\delta=((\sigma_0,\ldots,\sigma_{p-1}),\epsilon)\in\Wrp$ and $(e,s,v)\in\indexset$, we have $\delta(e,s,v)=(\epsilon(e),\sigma_{\epsilon(e)}(s),v)$.
We note that $|\Wrp|=\card{S_p} \cdot \card{S_2}^p=2^p p!$.

For $p\in\N$, we let $\Grp=S_2\times \Wrp$.
This $\Grp$ is the group of symmetries that we use to reorganize the formula in \cref{Francis} for the $p$th central moment into one more manageable to calculate.
The form of an element of our group is largely determined by the fact that when it permutes a partition $\cP$ of $\indexset$ by permuting the underlying elements of $\indexset$, these elements index the terms in the system of equations
\begin{align*}
\tau_{0,0,0}+\tau_{0,0,1} & = \tau_{0,1,0}+\tau_{0,1,1} \\
& \,\,\, \vdots  \\
\tau_{p-1,0,0}+\tau_{p-1,0,1} & = \tau_{p-1,1,0}+\tau_{p-1,1,1}
\end{align*}
that an assignment $\tau\in\CAs(\cP,=,\ell)$ must satisfy.  (See \cref{Colette} for $\CAs(\cP,=,\ell)$, which ultimately depends on \cref{Arnold}, where we find the equations.)
For $\gamma=(\digamma,((\sigma_0,\ldots,\sigma_{p-1}),\epsilon))\in\Grp$ and $(e,s,v)\in\indexset$, we have $\gamma(e,s,v)=(\epsilon(e),\sigma_{\epsilon(e)}(s),\digamma(v))$.
In this case, we say that {\it $\gamma$ uses the permutation $\epsilon$ to permute the equations, then uses $\sigma_e$ to permute the sides of equation $e$ for each $e \in [p]$, and then uses $\digamma$ to simultaneously permute the pair of places on every side of every equation.}
Since each element of $\Grp$ permutes $\indexset$, we regard $\Grp$ as a subgroup of $S_{\indexset}$.\footnote{Our group $\Grp$ here is closely related the group $\cW^{(p)}=(S_2\wr_{[2]} S_2) \wr_{[p]} S_p$ used in \cite[Notation 4.1]{Katz-Ramirez} to study the distribution of the autocorrelation demerit factor.  This $\cW^{(p)}$ is also regarded as a subgroup of $S_{\indexset}$.  If we let $i\colon \cW^{(p)} \to S_{\indexset}$ and $j \colon \Grp \to S_{\indexset}$ be the monomorphisms that give our identification of elements of $\cW^{(p)}$ and $\Grp$ with elements of $S_{\indexset}$, and if we let $\phi\colon \Grp \to \cW^{(p)}$ be the monomorphism $(\digamma,((\sigma_0,\ldots,\sigma_{p-1}),\epsilon)) \mapsto ((((\digamma,\digamma),\sigma_0),\ldots,((\digamma,\digamma),\sigma_{p-1})),\epsilon)$, then $i \circ \phi = j$, so that so that $\Grp$ can be regarded as a subgroup of $\cW^{(p)}$ via $\phi$.}
We  note that $|\Grp|=\card{S_2} \cdot |\Wrp|=2^{p+1} p!$, a fact that we shall use later when computing the size of orbits under this action.

Now we let our group $\Grp$ act on assignments.
In fact, we define the action of an arbitrary element of $S_{\indexset}$ on $\As([p])$, but since we identify $\Grp$ with a subgroup of $S_{\indexset}$, this defines an action of $\Grp$ on $\As([p])$.
If $p \in \N$ and $\pi\in S_{\indexset}$, then $\pi^*\colon \As([p]) \to \As([p])$ is the permutation given by $\pi^*(\tau)=\tau\circ\pi$, and one readily sees that the inverse of $\pi^*$ is $(\pi^{-1})^*$.
Now we show that our action preserves some important properties of assignments and transforms others in a straightforward way.
\begin{lemma}\label{Doreen}
Let $p,\ell \in \N$ and suppose that $\pi\in\Grp$ and $\cP\in\Coh(p)$.  Then
\begin{enumerate}[label=(\roman*)]
\item\label{Samantha} $\pi^*(\As([p]))=\As([p])$,
\item\label{Thomas} $\pi^*(\As([p],\ell))=\As([p],\ell)$,
\item\label{Ursula} $\pi^*(\As([p],=))=\As([p],=)$,
\item\label{Veronica} $\pi^*(\As([p],=,\ell))=\As([p],=,\ell)$,
\item\label{William} $\pi^*(\CAs(\pi(\cP)))=\CAs(\cP)$,
\item\label{Xavier} $\pi^*(\CAs(\pi(\cP),\ell))=\CAs(\cP,\ell)$,
\item\label{Yolanda} $\pi^*(\CAs(\pi(\cP)),=)=\CAs(\cP,=)$, and
\item\label{Zachary} $\pi^*(\CAs(\pi(\cP),=,\ell))=\CAs(\cP,=,\ell)$.
\end{enumerate}
Thus, the action of $\Grp$ on $\As([p])$ can be restricted to any of the following subsets: $\As([p],\ell)$, $\As([p],=)$, and $\As([p],=,\ell)$.
\end{lemma}
\begin{proof}
For each item, we need only prove that the left-hand set is contained in the right-hand one, and then we can get the opposite containment by invoking the one that we already proved with $\pi^{-1}$ in place of $\pi$ and $\pi(\cP)$ in place of $\cP$ (if $\cP$ appears) and then rearranging.  We have already noted that $\pi^*$ is a permutation of $\As([p])$, so \ref{Samantha} is clear.  From this, \ref{Thomas} follows because if $\tau \in \As([p],\ell)$, then \ref{Samantha} shows that $\pi^*(\tau) \in \As([p])$ and the image of $\pi^*(\tau)=\tau\circ\pi$ is contained in the image of $\tau$, which is contained in $[\ell]$.

To prove \ref{Ursula}, we define a map $\psi\colon \As([p]) \to \Z^{[p]}$ with $(\psi(\tau))(e)=\tau_{e 0 0}+\tau_{e 0 1}-\tau_{e 1 0}-\tau_{e 1 1}$ for each $e \in [p]$, and we also let $\phi$ be the unique isomorphism from $S_2$ to the cyclic multiplicative group $\Z^\times=(\{-1,1\},\cdot)$.  Note that if $\tau\in\As([p])$, then $\tau\in\As([p],=)$ if and only if $\psi(\tau)=0$ (where $0$ here means the zero function).  We let $\Grp$ act on $\Z^{[p]}$ as follows: if $\pi=(\digamma,((\sigma_0,\ldots,\sigma_{p-1}),\epsilon))\in\Grp$ and $u\in\Z^{[p]}$, then $\pi \cdot u$ is the element of $\Z^{[p]}$ with $(\pi\cdot u)(e)= \phi(\sigma_{\epsilon(e)}) u(\epsilon(e))$.
If $\pi\in\Grp$ and $\tau\in\As([p])$, then one shows that $\pi \cdot \psi(\tau)=\psi(\pi^*(\tau))$ by showing that the two functions have the same value for every input $e \in [p]$.
Therefore, if $\tau\in\As([p])$, then $\tau\in\As([p],=)$ if and only if $\psi(\tau)=0$, which is true if and only if $\pi\cdot\psi(\tau)=0$ for all $\pi\in\Grp$, which in turn is true if and only if $\psi(\pi^*(\tau))=0$ for all $\pi\in\Grp$, and this is true if and only if $\pi^*(\tau) \in \As([p],=)$ for all $\pi\in\Grp$.
This proves \ref{Ursula}, and then \ref{Veronica} follows from \ref{Thomas} and \ref{Ursula} because $\As([p],=,\ell)=\As([p],\ell)\cap\As([p],=)$.

To prove \ref{William}, let us write $\pi=(\digamma,((\sigma_0,\ldots,\sigma_{p-1}),\epsilon))$.
Suppose that $\tau\in\CAs(\pi(\cP))$.
Suppose that $(e,s,v),(e',s',v') \in \indexset$ with $v=v'$.
Then $\digamma(v)=\digamma(v')$, so the last components of $\pi(e,s,v)$ and $\pi(e',s',v')$ are identical.
Then our assumption that $\tau\in\CAs(\pi(\cP))$ tells us that $\pi(e,s,v)\equiv\pi(e',s',v') \pmod{\pi(\cP)}$ if and only if $\tau(\pi(e,s,v))=\tau(\pi(e',s',v'))$.
Thus, $(e,s,v)\equiv (e',s',v')\pmod{\cP}$ if and only if $(\pi^*(\tau))(e,s,v)=(\pi^*(\tau))(e',s',v')$.
Thus, $\pi^*(\tau)\in\CAs(\cP)$.

To prove \ref{Xavier}, \ref{Yolanda}, and \ref{Zachary}, let ``$\lozenge$'' stand for one of the decorations ``$\ell$'' or ``$=$'' or ``$=,\ell$''.
Let $\tau\in\pi^*(\CAs(\pi(\cP),\lozenge))$. By definition $\CAs(\pi(\cP),\lozenge) = \CAs(\pi(\cP))\cap\As([p],\lozenge)$, so $\tau\in\pi^*(\CAs(\pi(\cP))\cap\As([p],\lozenge)) \subseteq \pi^*(\CAs(\pi(\cP)))\cap\pi^*(\As([p],\lozenge))$, and thus by \ref{William} and the appropriate one of \ref{Thomas}--\ref{Veronica}, we have $\tau\in\CAs(\cP)\cap \As([p],\lozenge)=\CAs(\cP,\lozenge)$.
\end{proof}
Parts \ref{William}--\ref{Zachary} of \cref{Doreen} show that the action of $\pi\in\Grp$ on an assignment changes which partition the assignment coherently induces.
This means that we need to also consider the action of $\Grp$ on partitions of $\indexset$.
Partitions that lie in the same orbit under this action are considered to be isomorphic in some sense.
\begin{definition}[Coherently isomorphic partitions] Let $p \in \N$ and let $\cP,\cQ \in \Part(p)$.  We say that {\it $\cP$ and $\cQ$ are coherently isomorphic}, and write $\cP\Ccong\cQ$, to mean that there is some $\gamma\in\Grp$ such that $\gamma(\cP)=\cQ$.  The {\it coherent isomorphism class of $\cP$} is the set of all partitions that are coherently isomorphic to $\cP$.  In other words, the coherent isomorphism class of $\cP$ is the orbit of $\cP$ under the action of $\Grp$.
\end{definition}
\begin{example}\label{Odysseus}
Consider the partition $\cP \in \Coh(2)$ from \cref{Electra}.  Now consider the group element $\gamma=(\id,((\id,(01)),\id))$ in $\Grtwo$, where $\id$ denotes the identity permutation in $S_2$ and $(01)$ denotes the transposition in $S_2$.  Then let
\[
\cQ=\gamma(\cP)=\left\{\{(0,0,0),(1,1,0))\},\{(0,0,1),(1,1,1)\},\{(0,1,0),(1,0,0)\},\{(0,1,1),(1,0,1)\}\right\}, 
\]
so $\cP$ and $\cQ$ are coherently isomorphic.
\end{example}
Since coherent isomorphism classes are orbits under the action of $\Grp$, we define the stabilizer of a partition and then use the orbit-stabilizer counting formula to determine the size of a coherent isomorphism class.
\begin{definition}[Stabilizer of a partition] Let $p \in \N$ and let $\cP \in \Part(p)$.  The {\it stabilizer of $\cP$}, written $\Stab_{\Grp}(\cP)$, is the subgroup of $\Grp$ consisting of all $\gamma \in \Grp$ with $\gamma(\cP)=\cP$.
\end{definition}
\begin{lemma}
Let $p \in \N$ and let $\cP \in \Part(p)$.  Then the coherent isomorphism class of $\cP$ contains
\[
\frac{\card{\Grp}}{\card{\Stab_{\Grp}(\cP)}}=\frac{2^{p+1} p!}{\card{\Stab_{\Grp}(\cP)}}
\]
partitions.
\end{lemma}
\begin{example}\label{Penelope}
Consider the stabilizer of the partition $\cP\in\Coh(2)$ from \cref{Electra} under the action of $\Grtwo=S_2\times \Wrtwo$.
It is easy enough to check that the $\Stab_{\Grtwo}(\cP)=\{(\digamma,((\sigma_0,\sigma_1),\epsilon))\in\Grtwo : \digamma, \sigma_0, \sigma_1, \epsilon \in S_2; \sigma_0=\sigma_1\}$, so $\card{\Stab_{\Grtwo}(\cP)}=2^3$ and note that $|\Grtwo|=2^4$, so then the coherent isomorphism class of $\cP$ has $2^4/2^3=2$ partitions in it.
\cref{Odysseus} exhibits the other partition $\cQ$ that is isomorphic to $\cP$.
\end{example}
Most of the rest of this section is dedicated to showing that our group action organizes the coherently contributory partitions (i.e., elements $\cP$ in $\CCon(p)$ in the formula in \cref{Francis}) into orbits.
In order to see this, we need to show that our group action preserves the property of coherent contributoriness so that each orbit either contains only coherently contributory partitions or no coherently contributory partitions at all.
Since coherent contributoriness is the combination of coherence, GELOness, and satisfiability, we show that each of these three properties is preserved by our group action.
\begin{lemma}
Let $p\in\N$ and $\cP,\cQ\in\Part(p)$ with $\cP\Ccong\cQ$.
Then $\cP$ is coherent if and only if $\cQ$ is coherent.
\end{lemma}
\begin{proof}
Since $\cP$ and $\cQ$ are coherently isomorphic, there is some $\gamma=(\digamma,((\sigma_0,\ldots,\sigma_{p-1}),\epsilon))\in\Grp$ such that $\gamma(\cP)=\cQ$.
The partition $\cP$ is coherent if and only if $(e,s,0)$ and $(e',s',1)$ reside in different classes of $\cP$ for all $(e,s),(e',s') \in [p]\times[2]$.
This is true if and only if $\gamma(e,s,0)$ and $\gamma(e',s',1)$ reside in different classes of $\cQ$ for all $(e,s),(e',s') \in [p]\times [2]$.
This, in turn, is true if and only if $(\epsilon(e),\sigma_{\epsilon(e)}(s),\digamma(0))$ and $(\epsilon(e'),\sigma_{\epsilon(e')}(s'),\digamma(1))$ reside in different classes of $\cQ$ for all $(e,s),(e',s') \in [p]\times [2]$.
As we let $(e,s)$ (resp., $(e',s')$) run through $[p]\times[2]$, we see that $(f,t)=(\epsilon(e),\sigma_{\epsilon(e)}(s))$ (resp., $(f',t')=(\epsilon(e'),\sigma_{\epsilon(e')}(s'))$) also runs through all of $[p]\times[2]$.
So we know that $\cP$ is coherent if and only if $(f,t,\digamma(0))$ and $(f',t',\digamma(1))$ reside in different classes of $\cQ$ for all $(f,t),(f',t') \in [p]\times[2]$.
Since $\{\digamma(0),\digamma(1)\}=\{0,1\}$, we can then say that $\cP$ is coherent if and only if $\cQ$ is coherent.
\end{proof}
The following technical lemma is used to show that our group action preserves GELOness.
\begin{lemma}\label{Abraham}
Let $p\in\N$ and $\cP\in\Part(p)$.  For $e \in [p]$ and $\gamma=(\digamma,((\sigma_0,\ldots,\sigma_{p-1}),\epsilon))\in\Grp$, we have $(\gamma(\cP))_{\{\epsilon(e)\}}=\gamma(\cP_{\{e\}})$, so it follows that the partitions $(\gamma(\cP))_{\{\epsilon(e)\}}$ and $\cP_{\{e\}}$ have the same number classes of each size.
\end{lemma}
\begin{proof}
In view of how the permutation $\gamma\in\Grp$ acts on triples in the set $\indexset$, we have $\gamma(\{e\}\times\bindexset)=\{\epsilon(e)\}\times\bindexset$.
Therefore, $(\gamma(\cP))_{\{\epsilon(e)\}}$ and $\gamma((\cP)_{\{e\}})$ are both partitions of $\{\epsilon(e)\}\times\bindexset$.
Suppose that $(\epsilon(e),s,v),(\epsilon(e),t,w) \in \{\epsilon(e)\}\times\bindexset$.
Then $(\epsilon(e),s,v) \equiv (\epsilon(e),t,w) \pmod{(\gamma(\cP))_{\{\epsilon(e)\}}}$ if and only if  $(\epsilon(e),s,v) \equiv (\epsilon(e),t,w) \pmod{\gamma(\cP)}$, which is true if and only if $\gamma^{-1}(\epsilon(e),s,v) \equiv \gamma^{-1}(\epsilon(e),t,w) \pmod{\cP}$, which is true if and only if $\gamma^{-1}(\epsilon(e),s,v) \equiv \gamma^{-1}(\epsilon(e),t,w) \pmod{\cP_{\{e\}}}$ because $\gamma^{-1}(\epsilon(e),s,v)$ and $\gamma^{-1}(\epsilon(e),t,w)$ lie in $\{e\}\times\bindexset$.
Thus, $(\epsilon(e),s,v) \equiv (\epsilon(e),t,w) \pmod{(\gamma(\cP))_{\{\epsilon(e)\}}}$ if and only if $(\epsilon(e),s,v) \equiv (\epsilon(e),t,w) \pmod{\gamma(\cP_{\{e\}})}$, and so we see that $(\gamma(\cP))_{\{\epsilon(e)\}}$ and $\gamma(\cP_{\{e\}})$ are the same partition of $\{\epsilon(e)\}\times\bindexset$.  Since $\gamma$ permutes the elements of $\indexset$, for $P\in\cP_{\{e\}}$ we have $\card{P}=\card{\gamma(P)}$.  Thus, $(\gamma(\cP))_{\{\epsilon(e)\}}$ and $\cP_{\{e\}}$ have the same number classes of each size.
\end{proof}
\begin{corollary}\label{Beatrice}
Let $p\in\N$ and $\cP,\cQ\in\Part(p)$ with $\cP\Ccong\cQ$.
Then $\cP$ is GELO if and only if $\cQ$ is GELO.
\end{corollary}
\begin{proof}
Since $\cP$ and $\cQ$ are coherently isomorphic, there is some $\gamma=(\digamma,((\sigma_0,\ldots,\sigma_{p-1}),\epsilon))\in\Grp$ such that $\gamma(\cP)=\cQ$.
Since application of the permutation $\gamma$ of $\indexset$ does not change the size of classes, we know that $\cP$ is even if and only if $\cQ$ is even.
Also by \cref{Abraham}, for every $e\in[p]$ we have that $\cQ_{\{\epsilon(e)\}}$ is an even partition if and only if $\cP_{\{e\}}$ is an even partition.
Since $\epsilon$ permutes $[p]$, this shows that $\cP$ is GELO if and only if $\cQ$ is GELO.
\end{proof}
Our group action also preserves satisfiability.
\begin{lemma}\label{Eric}
Let $p,\ell\in\N$ and $\cP,\cQ\in\Coh(p)$ with $\cP\Ccong\cQ$.
Then $\card{\CAs(\cP,=,\ell)}=\card{\CAs(\cQ,=,\ell)}$ and $\cP\in\CSat(p)$ if and only if $\cQ\in\CSat(p)$.
\end{lemma}
\begin{proof}
Since $\cP\Ccong\cQ$, there is some $\gamma\in\Grp$ such that $\gamma(\cP)=\cQ$.
Thus, by \cref{Doreen}\ref{Zachary}, we know that $\gamma^*(\CAs(\gamma(\cP),=,\ell))=\CAs(\cP,=,\ell)$, and since $\gamma^*$ is a permutation of $\As([p])$, it follows that $\card{\CAs(\cQ),=,\ell)}=\card{\CAs(\cP,=,\ell)}$.  Furthermore, $\cP\in\CSat(p)$ if and only if $\card{\CAs(\cP,=,\ell)}>0$ for some $\ell \in \N$, and this is true if and only if $\card{\CAs(\cQ,=,\ell)}>0$ for some $\ell\in\N$, i.e., if and only if $\cQ\in\CSat(p)$.
\end{proof}
Now we can conclude that our group action preserves coherent contributoriness.
\begin{corollary}\label{Antelope}
Let $p\in\N$ and $\cP,\cQ\in\Part(p)$ with $\cP\Ccong\cQ$.
Then $\cP\in\CCon(p)$ if and only if $\cQ\in\CCon(p)$.
Furthermore, $\card{\CAs(\cP,=,\ell)}=\card{\CAs(\cQ,=,\ell)}$ for every $\ell\in\N$.
\end{corollary}
\begin{proof}
Since $\CCon(p)=\GELO(p)\cap\CSat(p)$, we combine \cref{Beatrice} and \cref{Eric}.
\end{proof}
Since our group action preserves coherent contributoriness, the set of coherently contributory partitions $\CCon(p)$ is a union of orbits under this action.
It will be helpful to name the set of all such orbits.
\begin{definition}[$\CIsom(p)$] 
Let $p\in\N$. We define $\CIsom(p)$ to be the set of coherent isomorphism classes of partitions in $\CCon(p)$.
\end{definition}
\begin{example}
Consider the partition $\cP$ from \cref{Electra}, which we showed to be in $\CCon(2)$ in \cref{Ajax}.
In \cref{Penelope} we found that the only other partition to which $\cP$ is coherently isomorphic is the partition $\cQ$ from \cref{Odysseus}.
\cref{Antelope} shows that $\cQ$ is also in $\CCon(2)$, and so $\{\cP,\cQ\}$ is an isomorphism class of partitions in $\CCon(2)$, i.e., $\{\cP,\cQ\} \in \CIsom(2)$.
\end{example}
Since \cref{Eric} shows that isomorphic partitions have the same count of assignments, we can attach the common value of that count to each isomorphism class.
\begin{definition}[$\CSols(\fP, \ell)$]\label{Baluga}
Let $p,\ell\in\N$. If $\fP\subseteq\Coh(p)$ such that all the partitions in $\fP$ are coherently isomorphic to each other, then we define $\CSols(\fP,\ell)$ to be the common value of $\card{\CAs(\cP,=,\ell)}$ for $\cP \in \fP$.  When we fix $\fP$ and regard $\CSols(\fP,\ell)$ as a function of $\ell$, we call it the {\it solution count function for $\fP$}.
\end{definition}
Now we are ready to prove a version of \cref{Francis} that organizes the summation by coherent isomorphism classes.
\begin{proposition}\label{Vito}
For $p,\ell\in\N$, we have
\[
\cmom\SSCC(f,g) = \sum_{\fP \in \CIsom(p)} \card{\fP}\CSols(\fP,\ell).
\]
\end{proposition}
\begin{proof}
From \cref{Francis}, we get
\begin{align*}
\cmom\SSCC(f,g) &= \sum_{\cP\in\CCon(p)} \card{\CAs(\cP,=,\ell)} \\
&= \sum_{\fP \in \CIsom(p)} \,\, \sum_{\cP\in\fP} \card{\CAs(\cP,=,\ell)}.
\end{align*}
\cref{Antelope} and \cref{Baluga} imply 
\begin{align*}
\cmom\SSCC(f,g)
& = \sum_{\fP \in \CIsom(p)} \,\, \sum_{\cP\in\fP} \CSols(\fP,\ell) \\
&= \sum_{\fP \in \CIsom(p)} \card{\fP} \CSols(\fP,\ell).\qedhere
\end{align*}
\end{proof}

\section{Finding isomorphism classes of coherently contributory partitions}\label{Theresa}

\cref{Vito} provides a way of calculating the central moments of the sum of squared magnitudes of crosscorrelation ($\SSCC$) provided that one can determine the isomorphism classes of coherently contributory partitions.
In this section, we describe a matrix-based method for finding all these isomorphism classes.
Thus, much of this section concerns how we encode partitions and their classes as matrices and vectors.
One should recall the definitions of coherent and atomized sets from \cref{Patricia}.
\begin{definition}[Display vector and matrix] Let $p \in \N$ and let $P$ be an atomized set that is a subset of $\indexset$.
The {\it display vector} of $P$ is the $p\times 1$ matrix (column vector) $v$ whose rows are indexed using the elements of $[p]$ and whose entries come from the set $\{0,1_r,1_b,-1_r,-1_b\}$, where the $0$ is considered uncolored while the elements with $r$ and $b$ subscripts are considered to be colored red or blue, respectively.  For each $e \in [p]$ our display vector $v$ has $v_e=0,1_r,1_b,-1_r,-1_b$ respectively when $P_{\{e\}}=\emptyset$, $\{(e,0,0)\}$, $\{(e,0,1)\}$, $\{(e,1,0)\}$, or $\{(e,1,1)\}$.
A {\it display matrix} of a collection $\cP$ of atomized sets in $\indexset$ is a $p\times\card{\cP}$ matrix whose columns are the display vectors of the classes of $\cP$ (arranged in any order).
\end{definition}

\cref{Alberto}\ref{Bertram} says that all classes in coherently contributory partitions are atomized, so every coherently contributory partition has a display matrix.
The correspondence between atomized sets and display vectors is bijective, but the correspondence between coherently contributory partitions and display matrices is not bijective, since one can permute the columns of a display matrix without changing which partition it corresponds to.
If we organize the display matrices into classes modulo permutation of columns, then the correspondence between these classes and the coherently contributory partitions is bijective.

Once $p$ becomes larger than $3$, it is difficult to run through all the different display matrices because they encode a great deal of information.
Our method starts with simpler matrices that lack some of the features (namely, color and signs) of display matrices.
\begin{definition}[$\mono$, monochrome vector and matrix]
The {\it monochromator}, written $\mono$, is the map that takes any matrix with entries in $\{0,1_r,1_b,-1_r,-1_b\}$ to the matrix with entries in $\{0,1,-1\}$ by replacing each entry that is $1_r$ or $1_b$ with $1$ and each entry that is $-1_r$ or $-1_b$ with $-1$.
For $p \in \N$, the {\it monochrome vector} of an atomized set $P$ in $\indexset$ is $\mono(v)$ where $v$ is the display vector of $P$.
A {\it monochrome matrix} for a collection $\cP$ of atomized sets in $\indexset$ is any $\mono(M)$ where $M$ is a display matrix of $\cP$.
\end{definition}

\begin{definition}[$\abs$, absolute vector and matrix]
The {\it absolutizer}, written $\abs$, is the map that takes any matrix with entries in $\{0,1,-1\}$ to the matrix with entries in $\{0,1\}$ by replacing each entry with its absolute value.
For $p \in \N$, the {\it absolute vector} of an atomized set $P$ in $\indexset$ is $\abs(\mono(v))$ where $v$ is the display vector of $P$.
An {\it absolute matrix} for a collection $\cP$ of atomized sets in $\indexset$ is any $\abs(\mono(M))$ where $M$ is a display matrix of $\cP$.
\end{definition}
Now we want to determine which display matrices correspond to coherently contributory partitions.
We build up to this by first examining which display matrices correspond to partitions and then move on to those that have additional properties such as coherence, GELOness, and satisfiability.
\begin{lemma}\label{Elizabeth}
Let $p,n\in\N$.
Let $M$ be a $p\times n$ matrix with entries in $\{0,1_r,1_b,-1_r,-1_b\}$.  Then $M$ is a display matrix for some partition in $\Part(p)$ if and only if $M$ satisfies the following conditions:
\begin{enumerate}[label=(\roman*)]
\item Each row of $M$ has one $1_r$, one $1_b$, one $-1_r$, one $-1_b$, and $n-4$ instances of $0$ as its entries (not necessarily in that order).
\item For each column of $M$, the sum of the absolute values (ignore color) of the entries in that column is strictly positive.
\end{enumerate}
\end{lemma}
\begin{proof}
Let $\cP$ be the collection of atomized sets that the columns of $M$ represent.
The second condition is equivalent to saying that these classes are nonempty.
The first condition is equivalent to saying that the classes are disjoint and their union is $\indexset$.
\end{proof}

\begin{definition}[Coherent display vector, coherent display matrix]
We say that a display vector is {\it incoherent} to mean that it has at least one red entry (i.e., $1_r$ or $-1_r$) and at least one blue entry (i.e., $1_b$ or $-1_b$).
A {\it coherent} display vector is a display vector that is not incoherent.
A nonzero display vector is said to be {\it red} (resp., {\it blue}) if it is coherent and its nonzero entries are all red (resp., blue), and then red (resp., blue) is said to be the {\it color} of that vector, while the zero vector has no color.
A {\it coherent} display matrix is a display matrix whose columns are all coherent display vectors; an {\it incoherent} display matrix is one that is not coherent.
A {\it red} (resp., {\it blue}) column of a coherent display matrix $M$ is a column of that matrix that is a red (resp., blue) display vector, and the {\it color} of that column is said to be red (resp., blue).
\end{definition}
Note that coherent display vectors are precisely the display vectors that represent coherent classes, and a red coherent display vector represents a $0$-coherent class while a blue coherent display vector represents a $1$-coherent class.

\begin{lemma}\label{Bruce}
Let $p,n\in\N$.
Let $M$ be a $p\times n$ matrix with entries in $\{0,1_r,1_b,-1_r,-1_b\}$.  Then $M$ is a display matrix for some partition in $\Coh(p)$ if and only if $M$ satisfies the following conditions:
\begin{enumerate}[label=(\roman*)]
\item Each row of $M$ has one $1_r$, one $1_b$, one $-1_r$, one $-1_b$, and $n-4$ instances of $0$ as its entries (not necessarily in that order).
\item For each column of $M$, the sum of the absolute values (ignore color) of the entries in that column is strictly positive.
\item $M$ is coherent.
\end{enumerate}
\end{lemma}
\begin{proof}
The conditions here are the same as those of \cref{Elizabeth} except for the added condition that $M$ be coherent, which is true if and only if all the classes of $\cP$ are coherent.
\end{proof}
Now that we have understood which display matrices correspond to coherent partitions, we look at which matrix property corresponds to GELOness.
\begin{lemma}\label{Geraldo}
Let $p,n\in\N$.
Let $M$ be a $p\times n$ matrix with entries in $\{0,1_r,1_b,-1_r,-1_b\}$.  Then $M$ is a display matrix for some partition in $\GELO(p)$ if and only if $M$ satisfies the following conditions:
\begin{enumerate}[label=(\roman*)]
\item Each row of $M$ has one $1_r$, one $1_b$, one $-1_r$, one $-1_b$, and $n-4$ instances of $0$ as its entries (not necessarily in that order).
\item For each column of $M$, the sum of the absolute values (ignore color) of the entries in that column is strictly positive and even.
\end{enumerate}
\end{lemma}
\begin{proof}
The conditions here are the sames as those of \cref{Elizabeth} except for the added parity condition.
So if $\cP$ is the collection of atomized sets of $\indexset$ represented by the columns of $M$, then the conditions here are equivalent to saying that $\cP$ is an even partition.
But a partition consisting of atomized classes is GELO if and only if it is even.
\end{proof}
Now we want to characterize which display matrices correspond to coherently satisfiable partitions.
This requires us to build some further linear-algebraic machinery.
\begin{definition}[colored vectors, weights, and distances]
A {\it colored} vector is a vector whose entries are in $\C$ and each entry is given the color red or blue (but not both).
Two colored vectors $v=(v_1,\ldots,v_m)$ and $w=(w_1,\ldots,w_n)$ are said to have the {\it same color scheme} if and only if $m=n$ and $v_j$ has the same color as $w_j$ for all $j \in \{1,\ldots,m\}$.
The sum or difference of two vectors with the same color scheme is endowed with the same color scheme as the two vectors being combined.
The {\it colored weight} of a vector $v$ is the ordered pair $(a,b)$ where $a$ is the number of nonzero red entries in $v$ and $b$ is the number of nonzero blue entries in $v$.
The {\it colored distance} between two colored vectors $v$ and $w$ having the same color scheme is equal to the colored weight of $v-w$.
A {\it properly colored vector} is a colored vector such that no two entries have both the same color and the same value.
\end{definition}
We can characterize coherent satisfiability of partitions via certain homogeneous systems of equations related to the matrices that represent the partitions.
\begin{lemma}\label{Timothy}
Let $p,\ell\in\N$, let $\cP \in \Coh(p)$, and let $M$ be a display matrix for $\cP$.
Index the rows of $M$ with the elements of $[p]$ and index the columns of $M$ with the classes of $\cP$, where class $P$ indexes the column that is the display vector of $P$.
We write each $\card{\cP}$-tuple $a=(a_P)_{P \in \cP}$ in $[\ell]^{\cP}$ as a column vector whose entries are indexed by the classes of $\cP$ arranged in the same order that they are for the columns of $M$.
Consider all such $\card{\cP}$-tuples to be colored so that for each $P\in\cP$, the color of entry $a_P$ is the same as the color of the column of $M$ indexed by $P$.
Then $\card{\CAs(\cP,=,\ell)}$ is equal to the number of such properly colored $\card{\cP}$-tuples $a=(a_P)_{P \in \cP}$ in $[\ell]^{\cP}$ such $M a=0$ (ignore colors of entries of $M$ and $a$ when computing $M a$).
Therefore, $\cP$ is coherently satisfiable if and only if there is there some properly colored $\card{\cP}$-tuple $a=(a_P)_{P \in \cP}$ in $\N^{\cP}$ such that $M a=0$ and the color of entry $a_P$ is the same as the color of the column of $M$ indexed by $P$ for every $P \in \cP$.
\end{lemma}
\begin{proof}
Note that every column of $M$ is either red or blue since no column is a zero vector by \cref{Bruce}.
Let $A$ be the set of all properly colored $a \in [\ell]^{\cP}$ such that $M a=0$ and where the color of the $P$th entry of $a$ is the same as the color of the $P$th column of $M$ for every $P \in \cP$.
We write both $a_P$ (as in the statement) and $a(P)$ to mean the $P$th coordinate of $a$.
The proper coloring condition forces distinct red entries of $a$ to have different values, which is equivalent to saying that $a_P\not=a_Q$ whenever $P$ and $Q$ are distinct $0$-coherent classes.
Likewise, the proper coloring condition forces distinct blue entries of $a$ to have different values, which is equivalent to saying that $a_P\not=a_Q$ whenever $P$ and $Q$ are distinct $1$-coherent classes.
Throughout this proof, we let $\phi\colon \indexset \to \cP$ be the quotient map for the partition $\cP$, i.e., $\phi(e,s,v)$ is equal to the class $P \in \cP$ such that $(e,s,v) \in P$.

We want to define a map $\Phi\colon A \to \CAs(\cP,=,\ell)$ where $\Phi(a) \in \As([p])$ such that $(\Phi(a))(e,s,v)=a(\phi(e,s,v))$ for each $a \in A$ and each $(e,s,v) \in \indexset$, so we need to check that our $\Phi(a)$ really is in $\CAs(\cP,=,\ell)$.
Clearly, our map $\Phi(a)\in\As([p],\ell)$ because $a(\phi(e,s,v)) \in [\ell]$ for every $(e,s,v) \in \indexset$.
By the way display vectors and matrices are defined, the $e$th row of the display matrix $M$ for $\cP$ has a $1_r$, a $1_b$, a $-1_r$, and a $-1_b$ in the columns indexed by $\phi(e,0,0)$, $\phi(e,0,1)$, $\phi(e,1,0)$, and $\phi(e,1,1)$, respectively, and all the rest of the entries of this row are zeros.
Thus, the $e$th entry of $M a$ is equal to
\[
a(\phi(e,0,0)) +a(\phi(e,0,1))-a(\phi(e,1,0))-a(\phi(e,1,1)),
\]
which equals $0$ since $M a=0$.
Thus, by the way we defined $\Phi$, we have
\[
(\Phi(a))(e,0,0)+(\Phi(a))(e,0,1)-(\Phi(a))(e,1,0)-(\Phi(a))(e,1,0) = 0,
\]
and since this is true for all $e \in [p]$, we conclude that $\Phi(a) \in \As([p],=)$.
We let $\cQ$ be the partition in $\Coh(p)$ that is coherently induced by $\Phi(a)$.
For any $v \in [2]$ and any $(e,s,v),(e',s',v) \in \indexset$, we have $(e,s,v) \equiv (e',s',v) \pmod{\cQ}$ if and only if $(\Phi(a))(e,s,v)=(\Phi(a))(e',s',v)$, that is, if and only if $a(\phi(e,s,v))=a(\phi(e',s',v))$.
Since $\phi(e,s,v)$ and $\phi(e',s',v)$ are both $v$-coherent classes of the coherent partition $\cP$, the columns of $M$ that correspond to these two classes must be the same color, and therefore $a(\phi(e,s,v))=a(\phi(e',s',v))$ if and only if $\phi(e,s,v)=\phi(e',s',v)$, i.e., if and only if $(e,s,v) \equiv (e',s',v) \pmod{\cP}$.  So we see that both $\cP$ and $\cQ$ are coherent partitions and $(e,s,v) \equiv (e',s',v) \pmod{\cQ}$ if and only if $(e,s,v)\equiv(e',s',v) \pmod{\cP}$ for every $v \in [2]$ and $(e,s,v),(e',s'v) \in \indexset$.  Thus, $\cQ=\cP$, and then we can conclude that $\Phi(a)$ coherently induces $\cP$, i.e., $\Phi(a) \in \CAs(\cP)$.
Putting together what we have shown, we obtain $\Phi(a) \in \As([p],\ell) \cap \As([p],=) \cap \CAs(\cP)=\CAs(\cP,=,\ell)$.
So our map $\Phi\colon A \to \CAs(\cP,=,\ell)$ is indeed defined.

We want to define a map $\Psi\colon\CAs(\cP,=,\ell) \to A$ where $\Psi(\tau) \in [\ell]^{\cP}$ such that for each $P \in \cP$ we define $\Psi(\tau)(P)$ to be the common value of $\tau_{e,s,v}$ for every $(e,s,v) \in P$ (there is such a common value because $\tau$ coherently induces $\cP$ since $\tau\in\CAs(\cP,=,\ell)$) and we color $\Psi(\tau)(P)$ with the color of the column of the matrix $M$ that is indexed by $P$.
We can be sure that all the entries of the vector $\Psi(\tau)$ lie in $[\ell]$ because $\tau_{e,s,v} \in [\ell]$ for all $(e,s,v)\in\indexset$ because $\tau \in \CAs(\cP,=,\ell)$.
We also note that we have colored the elements of $\Psi(\tau)$ with the correct color scheme for them to be in $A$.
We write both $\Psi(\tau)_P$ (as in the statement) and $(\Psi(\tau))(P)$ to mean the $P$th coordinate of $\Psi(\tau)$.
Now we need to check that $M\Psi(\tau)=0$ for all $\tau\in\CAs(\cP,=,\ell)$.
For $e \in [p]$, the $e$th entry of $M\Psi(\tau)$ is the product of the $e$th row of $M$ with the column vector $\Psi(\tau)$.
By the way display vectors and matrices are defined, the $e$th row of the display matrix $M$ for $\cP$ has a $1_r$, a $1_b$, a $-1_r$, and a $-1_b$ in the columns indexed by $\phi(e,0,0)$, $\phi(e,0,1)$, $\phi(e,1,0)$, and $\phi(e,1,1)$, respectively, and all the rest of the entries of this row are zeros.
Thus, the $e$th entry of $M \Psi(\tau)$ is equal to
\[
(\Psi(\tau))(\phi(e,0,0)) +(\Psi(\tau))(\phi(e,0,1))-(\Psi(\tau))(\phi(e,1,0))-(\Psi(\tau))(\phi(e,1,1)),
\]
and since it is immediate that $(e,s,v)\in\phi(e,s,v)$ for all $s,v \in [2]$ by the definition of $\phi$, they way we defined $\Psi(\tau)$ makes the $e$th entry of $M\Psi(\tau)$ equal to
\[
\tau_{e,0,0} +\tau_{e,0,1}-\tau_{e,10}-\tau_{e,1,1},
\]
which is $0$ because $\tau\in\CAs(\cP,=,\ell)$.
So we have proved that every entry of $M\Psi(\tau)$ is $0$, and so $M\Psi(\tau)=0$.
Now let $P,Q \in \cP$ be distinct classes such that $\Psi(\tau)_P$ and $\Psi(\tau)_Q$ have the same color.
This means that the columns of $M$ corresponding to $P$ and $Q$ have the same color.
So there is some $v \in [2]$ such that both $P$ and $Q$ are $v$-coherent.
Choose some $(e,s,v) \in P$ and $(e',s',v) \in Q$.
Since $\tau$ coherently induces $\cP$ and $P$ and $Q$ are distinct $v$-coherent classes of $\cP$, we conclude that $\tau_{e,s,v} \not=\tau_{e',s',v}$, and so $(\Psi(\tau))_P\not=(\Psi(\tau))_Q$.
Thus, $\Psi(\tau)$ is properly colored.
All in all, we have shown that $\Psi(\tau) \in A$, and so our map $\Psi\colon\CAs(\cP,=,\ell) \to A$ is defined.

Now suppose that $a \in A$, and notice that for $P \in \cP$, the $P$th entry of $\Psi(\Phi(a))$ is equal to the common value of $(\Phi(a))_{e,s,v}$ for all $(e,s,v) \in P$, which in turn is equal to $a(\phi(e,s,v))=a(P)$.
So we conclude that $\Psi(\Phi(a))=a$ for all $a \in A$.
On the other hand, if we let $\tau\in\CAs(\cP,=,\ell)$ and $(e,s,v) \in \indexset$, then $(\Phi(\Psi(\tau))_{e,s,v}=\Psi(\tau)(\phi(e,s,v))$, and since $(e,s,v)\in\phi(e,s,v)$ by the definition of $\phi$, we see that $(\Phi(\Psi(\tau))_{e,s,v}=\tau_{e,s,v}$.
So we conclude that $\Phi(\Psi(\tau))=\tau$ for all $\tau\in\CAs(\cP,=,\ell)$.
Thus $\Phi$ and $\Psi$ are inverses of each other, and so both maps are bijective, and so $\card{A}=\card{\CAs(\cP,=,\ell)}$.

The claim about coherent satisfiability of the lemma follows immediately from the earlier claim.
\end{proof}
The following result leads to \cref{Oswald}, which provides a simple criterion for the existence of solutions that verify satisfiability in \cref{Timothy}.
\begin{lemma}\label{Henrietta}
Suppose that $F$ is a subfield of $\C$ and that $M$ is an $m\times n$ matrix with entries in $F$ whose row sums are all zero.
Suppose that each column of $M$ has a color, either red or blue.
Let $M'$ be the reduced row echelon form of $M$, and let the $j$th column of $M'$ have the same color as the $j$th column of $M$ for each $j \in \{1,2,\ldots,n\}$.
Each entry of $M'$ is given the color of its column, so that the rows of $M'$ are colored vectors that all share the same color scheme.
Every solution $x$ of the homogeneous linear system $M x=0$ is given the same color scheme as the rows of $M'$.
Then $M x=0$ has a properly colored solution whose entries are in $F$ if and only if $M'$ meets the following conditions.
\begin{enumerate}[label=(\roman*)]
\item No row of $M'$ has colored weight $(2,0)$.
\item No row of $M'$ has colored weight $(0,2)$.
\item No pair of rows of $M'$ have colored distance $(2,0)$.
\item No pair of rows of $M'$ have colored distance $(0,2)$.
\end{enumerate}
\end{lemma}
\begin{proof}
Assume that $M'$ meets the four conditions enumerated in the statement of this lemma.
Since the row sums in $M'$ are zero (because the row sums in $M$ are zero), no difference between two rows of $M'$ can have colored weight equal to $(1,0)$, $(0,1)$, $(2,0)$, $(0,2)$.
Let $E$ be the union of $\{0,1,-1\}$ and the set of all entries of $M'$ and let $D$ be the set of all pairwise differences of elements in $E$ (including differences of the form $d-d=0$).
Let $s=\min_{d \in D \smallsetminus\{0\}} |d|$ and $\ell=\max_{d \in D} |d|$.
Solutions $x$ of $M x=0$ are identical to solutions of $M x'=0$, and the row reduced echelon form $M'$ gives us a way of writing each pivot variable as a linear combination of the free variables for solutions where the coefficients in the linear combination are opposites of elements of $E$.
Therefore, the difference between any two distinct variables in the solution set is a linear combination of free variables with the coefficients in the linear combination being elements of $D$.
Furthermore, our conditions on the colored weights of rows and differences of rows guarantees that if the columns corresponding to the two variables are of the same color, then their difference is a linear combination that has at least one nonzero coefficient.
Let $q$ be a rational number (which lies in $F$ since $\Q$ is $F$'s prime subfield) with $q > (c+1) \ell /s$ where $c$ is the number of columns in $M'$; this makes $q > 1$.
Let $a$ be a colored vector with the same scheme as the rows of $M'$, and for coordinates in this vector that correspond to free variables, assign values of $q, q^2, \ldots, q^k$, where $k$ is the number of free variables.
Let the coordinates of $a$ corresponding to pivot variables take the values they must take so that $M' a=0$ (and so $M a=0$).
The entries of $a$ are in $F$ because the free variables are assigned powers of $q \in \Q \subseteq F$, and the pivot variables are assigned $F$-linear combinations of the powers of $q$ (since $M'$ has entries in $F$).
The difference between any two coordinates of $a$ having the same color is a sum of the form $d=\sum_{i=1}^j d_i q^i$ with $1 \leq j \leq k$ and $d_1,\ldots,d_j \in D$ and $d_j\not=0$.
And then $|d_j q^j| \geq s q^j > (c+1)\ell q^{j-1} > k \ell q^{j-1} \geq j \ell q^{j-1} > \sum_{i=1}^{j-1} |d_i q^i| \geq \left|\sum_{i=1}^{j-1} d_i q^i\right|$, and so $d\not=0$.
So we have shown that $M x=0$ has a properly colored solution $a$ whose entries are in $F$.

If any of the four conditions on $M'$ listed in the statement of this lemma does not hold, then there is some vector $v$ in the row space of $M'$ whose entries sum to zero and whose colored weight is $(2,0)$ or $(0,2)$.
So there is some nonzero element $e$ of $F$ such that $v$ has one $e$, one $-e$ and all other entries equal to $0$, and the $e$ and $-e$ have the same color.
But then this forces any solution of $M' x=0$ (hence any solution of $M x=0$) to not be properly colored, since it requires two different coordinates with the same color to be assigned the same value.
\end{proof}

\begin{corollary}\label{Oswald}
Suppose that $M$ is an $m\times n$ matrix with entries in $\Z$ whose row sums are all zero.
Suppose that each column of $M$ has a color, either red or blue.
Let $M'$ be the reduced row echelon form of $M$, and let the $j$th column of $M'$ have the same color as the $j$th column of $M$ for each $j \in \{1,2,\ldots,n\}$.
Each entry of $M'$ is given the color of its column, so that the rows of $M'$ are colored vectors that all share the same color scheme.
Every solution $x$ of the homogeneous linear system $M x=0$ is given the same color scheme as the rows of $M'$.
Then $M x=0$ has a properly colored solution whose entries are in $\N$ if and only if $M'$ meets the following conditions.
\begin{enumerate}[label=(\roman*)]
\item No row of $M'$ has colored weight $(2,0)$.
\item No row of $M'$ has colored weight $(0,2)$.
\item No pair of rows of $M'$ have colored distance $(2,0)$.
\item No pair of rows of $M'$ have colored distance $(0,2)$.
\end{enumerate}
\end{corollary}
\begin{proof}
\cref{Henrietta} shows that we have a properly colored solution of $M x=0$ whose coordinates are in $\Q$ if and only if the enumerated conditions above are met.
Every properly colored solution with coordinates in $\Q$ can be turned into a properly colored solution with coordinates in $\Z$ by scaling (because the system is homogeneous) and thence to a properly colored solution with coordinates in $\N$ by translating by some integer multiple of the all-one vector (since the row sums of $M$ are zero).
\end{proof}

Now combine \cref{Bruce}, \cref{Timothy}, and \cref{Oswald} to obtain the following characterization of the display matrices representing coherently satisfiable partitions.
\begin{lemma}\label{Sally}
Let $p,n\in\N$.
Let $M$ be a $p\times n$ matrix with entries in $\{0,1_r,1_b,-1_r,-1_b\}$.
Then $M$ is a display matrix for some partition in $\CSat(p)$ if and only if $M$ satisfies the following conditions:
\begin{enumerate}[label=(\roman*)]
\item Each row of $M$ has one $1_r$, one $1_b$, one $-1_r$, one $-1_b$, and $n-4$ instances of $0$ as its entries (not necessarily in that order).
\item For each column of $M$, the sum of the absolute values (ignore color) of the entries in that column is strictly positive.
\item $M$ is coherent.
\item If $M'$ is the reduced row echelon form of $M$ whose columns are given the same sequence of colors as the columns of $M$ (using the coherence of $M$ from the previous condition) and each entry of $M'$ is given the color of the column it resides in, then no row of $M'$ has colored weight $(2,0)$ or $(0,2)$, and no pair of rows of $M'$ have colored distance $(2,0)$ or $(0,2)$.
\end{enumerate}
\end{lemma}
\begin{proof}
By \cref{Bruce}, the first three conditions are necessary and sufficient for $M$ to be a display matrix for some partition $\cP$ in $\Coh(p)$.
Once these three are granted, then by \cref{Timothy} and \cref{Oswald}, the fourth condition is necessary and sufficient for $\cP$ to be coherently satisfiable.
\end{proof}
And now we synthesize Lemmas \ref{Geraldo} and \ref{Sally} to obtain the following characterization of the display matrices that represent coherently contributory partitions.
\begin{lemma}\label{Concetta}
Let $p,n\in\N$.
Let $M$ be a $p\times n$ matrix with entries in $\{0,1_r,1_b,-1_r,-1_b\}$.
Then $M$ is a display matrix for some partition in $\CCon(p)$ if and only if $M$ satisfies the following conditions:
\begin{enumerate}[label=(\roman*)]
\item Each row of $M$ has one $1_r$, one $1_b$, one $-1_r$, one $-1_b$, and $n-4$ instances of $0$ as its entries (not necessarily in that order).
\item For each column of $M$, the sum of the absolute values (ignore color) of the entries in that column is strictly positive and even.
\item $M$ is coherent.
\item If $M'$ is the reduced row echelon form of $M$ whose columns are given the same sequence of colors as the columns of $M$ (using the coherence of $M$ from the previous condition) and each entry of $M'$ is given the color of the column it resides in, then no row of $M'$ has colored weight $(2,0)$ or $(0,2)$, and no pair of rows of $M'$ have colored distance $(2,0)$ or $(0,2)$.
\end{enumerate}
\end{lemma}
\begin{proof}
Since $\CCon(p)=\GELO(p)\cap \CSat(p)$, we just take the conjunction of the conditions from Lemmas \ref{Geraldo} and \ref{Sally}.
\end{proof}
For each $p \in \N$, we would like to find a set of representatives for the coherent isomorphism classes of coherently contributory partitions of $\indexset$, i.e., a collection of partitions with exactly one partition per element of $\CIsom(p)$.
We shall use display matrices to find these representatives.
We note that permutation of columns does not change which partition such a display matrix represents, so when we search for display matrices for our representatives, we are looking for matrices up to column permutation.
Furthermore, since coherent isomorphism classes of partitions are orbits under the action of the group $\Grp$, it will be helpful to see how the action of some $\gamma\in\Grp$ on a partition $\cP$ changes a display matrix $M$ of $\cP$ to a display matrix of $\gamma(\cP)$.
Recall that $\Grp=S_2 \times \Wrp$ where $\Wrp=S_2 \wr_{[p]} S_p$, so that an element of $\Grp$ is of the form $\gamma=(\digamma,((\sigma_0,\ldots,\sigma_{p-1}),\epsilon))$ with $\digamma,\sigma_0,\ldots,\sigma_{p-1} \in S_2$ and $\epsilon \in S_p$.
This $\gamma$ maps an $(e,s,v)\in\indexset$ to $\gamma(e,s,v)=(\epsilon(e),\sigma_{\epsilon(e)}(s),\digamma(v))$.
Recall that we say that {\it $\gamma$ uses the permutation $\epsilon$ to permute the equations, then uses $\sigma_e$ to permute the sides of equation $e$ for each $e \in [p]$, and then uses $\digamma$ to simultaneously permute the pair of places on every side of every equation.}
If $M$ is a display matrix of some partition $\cP$ of $\indexset$ where all classes of $\cP$ are atomized, and $\gamma$ is the element of $\Grp$ just described, then one obtains a display matrix of $\gamma(\cP)$ from $M$ by permuting the rows of $M$ using $\epsilon$ to get some $M'$, then for each $e \in [p]$ we negate the $e$th row of $M'$ if and only if $\sigma_e=(0 1)$, and finally we change the color of all colored elements (red becomes blue and blue becomes red) if and only if $\digamma=(0 1)$ (but do nothing if $\digamma$ is the identity permutation) to obtain a display matrix for $\gamma(\cP)$.
So, searching for representatives of coherent isomorphism classes of partitions in $\CCon(p)$ is tantamount to finding a set of representatives of the display matrices that satisfy the conditions of \cref{Concetta} under the action that allows us to permute columns, permute rows, negate each row independently of the others, and swap red and blue colors throughout the entire matrix.
A sketch of a procedure for finding all these matrices follows.
\begin{enumerate}
\item Consider all the matrices with $p$ rows and entries from $\{0,1\}$ where each row sum is $2$ and each column sum is a positive even integer, and let $A$ be a set of representatives of orbits of the matrices under the action of column permutation.
\item For each matrix in $A$, we can attach signs to the $1$ entries so that each row has one instance of $+1$ and one instance of $-1$.  There are $2^p$ ways to do this for each matrix in $A$, and if we collect all such matrices into a set of $2^p\card{A}$ matrices, then let $S$ be a set of representatives of orbits of these signed matrices under the action of column permutation.
\item For each pair of matrices $(M_0,M_1)$ in $S\times S$, color all the nonzero elements of $M_0$ red and all the elements of $M_1$ blue and form the matrix $M=[M_0|M_1]$ by juxtaposing $M_1$ to the right of $M_0$; call the set of all such matrices produced this way $T$.  From all these matrices, let $U$ be a set of representatives of orbits of the matrices in $T$ under the group action that allows us to arbitrarily permute columns, permute rows, negate any subset of rows, and globally swap red and blue colors for all nonzero entries of the matrix.
Let $D$ be the set of all matrices $M$ in $U$ such that the reduced row echelon form $M'$ of $M$ (with every entry in the $j$th column of $M'$ given the same color as the $j$th column of $M$ for each column index $j$) has the property that no row of $M'$ has colored weight $(2,0)$ or $(0,2)$ and no two rows of $M'$ have colored distance $(2,0)$ or $(0,2)$.  This set $D$ will have one and only one display matrix for each class $\fC$ in $\CIsom(p)$ with that matrix representing one of the coherently contributory partitions in $\fC$.
\end{enumerate}
The last claim in our procedure about $D$ is true because (i) the method of construction guarantees that all matrices in $D$ conform to the properties enumerated in \cref{Concetta}; (ii) for every matrix $M$ that has the properties enumerated in \cref{Concetta}, there must be a matrix in $T$ that has the same columns as $M$ but perhaps in a different order (so as to place all the columns with red entries to the left of all the columns with blue entries); and (iii) two matrices represent the same coherent isomorphism class of $\CIsom(p)$ if and only if they lie in the same orbit of the group action described in the procedure.
\begin{example}\label{Hector}
We use our procedure to find display matrices for a set of representatives of the classes in $\CIsom(2)$.
First we construct the set $A$ of representatives modulo column permutation of all matrices with $2$ rows and entries from $\{0,1\}$ where each row sum is $2$ and each column sum is a positive even integer.
This means that the sum of all entries is $4$, and every column sum is $2$, so the only matrix in $A$ is
\[
L=\begin{bmatrix} 1 & 1 \\ 1 & 1 \end{bmatrix}.
\]

Then we construct the set $S$ of matrices obtained by changing one entry in each row of $L$ from a $+1$ to a $-1$ and we choose representatives modulo column permutation, and so we obtain $S=\{L_1,L_2\}$, where
\begin{align*}
L_1 =\begin{bmatrix} +1 & -1 \\ +1 & -1 \end{bmatrix} \text{ and } L_2 =\begin{bmatrix} +1 & -1 \\ -1 & +1 \end{bmatrix}.
\end{align*}

Now we form the set $T$ of all four matrices of the form $[M_0|M_1]$ with $(M_0,M_1) \in S\times S$ and with the nonzero entries in $M_0$ colored red and the nonzero entries in $M_1$ colored blue.
Then we form a set $U$ of representatives of matrices in $T$ modulo the group action that allows us to arbitrarily permute columns, permute rows, negate any subset of rows, and globally swap red and blue colors for all entries of the matrix.  Our $U=\{K_1,K_2\}$ with
\[
K_1 = \begin{bmatrix}
+1_r & -1_r & +1_b & -1_b \\
+1_r & -1_r & +1_b & -1_b
\end{bmatrix}
\text{ and }
K_2 = \begin{bmatrix}
+1_r & -1_r & +1_b & -1_b \\
+1_r & -1_r & -1_b & +1_b
\end{bmatrix}.
\]
Finally, we let $D$ be the set containing every matrix in $U$ whose reduced row echelon form (with the entries of any given column in the reduced matrix given the same color as the corresponding column of the original) lacks any row of colored weight $(2,0)$ or $(0,2)$ and any pair of rows with colored distance $(2,0)$ or $(0,2)$.
The row reduced echelon forms of $K_1$ and $K_2$ are
\[
 \begin{bmatrix}
+1_r & -1_r & +1_b & -1_b \\
 0_r &  0_r &  0_b &  0_b
\end{bmatrix}
\text{ and }
\begin{bmatrix}
+1_r & -1_r &  0_b &  0_b \\
 0_r &  0_r & +1_b & -1_b
\end{bmatrix},
\]
respectively.
Since the former reduced row echelon form has no row with colored weight $(2,0)$ or $(0,2)$ and no two rows with colored distance $(2,0)$ or $(0,2)$, we conclude that $K_1 \in D$.
But the latter reduced row echelon form has a row of colored weight $(2,0)$, so $K_2 \not\in D$.
So $D=\{K_1\}$, and so there is precisely one class in $\CIsom(2)$.
This $K_1$ is the display matrix for the partition $\cP$ from \cref{Electra}, so that partition is a representative of the sole coherent isomorphism class in $\CIsom(2)$.
\end{example}

\section{Positivity of moments}\label{Portia}

This section is dedicated to a proof of \cref{Priscilla}, which states that all moments of the crosscorrelation demerit factor are nonnegative (and usually strictly positive unless the sequences are very short or if we are considering the first central moment, which is always zero).
Since $\mu^\ell_{p,(f,g)} \CDF(f,g) = \mu^\ell_{p,(f,g)} \SSCC(f,g) /\ell^{2 p}$ (when $\ell>0$) by \eqref{Celeste}, we see that $\mu^\ell_{p,(f,g)} \CDF(f,g)$ is nonnegative (resp., strictly positive) if and only if $\mu^\ell_{p,(f,g)} \SSCC(f,g)$ is nonnegative (resp., strictly positive).
So for convenience, we prove our claims here with $\SSCC$ in place of $\CDF$, and then \cref{Priscilla} follows immediately.
\begin{theorem}
If $\ell \in \N$ and $p \in \Z_+$, then $\mu^\ell_{p,(f,g)} \SSCC(f,g)$ is nonnegative.  Moreover, if (i) $p=1$, (ii) if $p$ is odd and $\ell \leq 2$, or (iii) if $p$ is even and $\ell \leq 1$, then $\mu^\ell_{p,(f,g)} \SSCC(f,g)=0$; otherwise it is strictly positive.
\end{theorem}
\begin{proof}
\cref{Francis} expresses $\mu^\ell_{p,(f,g)} \SSCC(f,g)$ as a sum of cardinalities of sets, so it cannot be negative.
We have $\mu^\ell_{1,(f,g)} \SSCC(f,g)=0$ because the first central moment is always zero.

Suppose that $\cP$ is a coherently contributory partition.
Then \cref{Alberto}\ref{David} shows that $\cP$ has at least four classes in it (and if $p$ is odd, then it has at least $6$ classes in it).
In view of coherence, there must be some $v \in [2]$ such that $\cP$ has at least two classes (and at least three classes if $p$ is odd) that are subsets of $[p]\times[2]\times\{v\}$.
Then \cref{Wayne} shows that $\CAs(\cP,\ell)=\emptyset$ when $\ell < 2$ (and if $p$ is odd, when $\ell < 3$).
Therefore $\card{\CAs(\cP,=,\ell)}=0$ when $\ell < 2$ (if $p$ is odd, when $\ell < 3$).
Thus, \cref{Francis} shows that $\mu^\ell_{p,(f,g)} \SSCC(f,g)=0$ when $\ell \leq 1$ (and also when $\ell\leq 2$ if $p$ is odd).

If $p$ is even and $\ell \geq 2$, let $\cQ=\{P,Q,R,S\}$, where $P=[p]\times\{0\}\times\{0\}$, $Q=[p]\times\{0\}\times\{1\}$, $R=[p]\times\{1\}\times\{0\}$, and $S=[p]\times\{1\}\times\{1\}$, which is a coherent and GELO partition of $\indexset$.  Then let $\tau\colon\indexset \to \N$ be defined by $\tau_\alpha=0$ if $\alpha \in P \cup S$ and $\tau_\alpha=1$ if $\alpha \in Q\cup R$.  It is easy to check that $\tau \in \CAs(\cQ,=,2) \subseteq \CAs(\cQ,=,\ell) \subseteq \CAs(\cQ,=)$, which makes $\cQ$ coherently satisfiable (and hence coherently contributory since it is GELO), and makes $\card{\CAs(\cQ,=,\ell)}$ strictly positive.
This makes $\mu^\ell_{p,(f,g)} \SSCC(f,g)$ strictly positive by \cref{Francis}.

If $p$ is odd with $p>1$ and $\ell \geq 3$, let $\cR=\{A,B,C,D,E,F\}$ where $A=\{0,1,\ldots,p-2\}\times\{0\}\times\{0\}$, $B=\{0,1,\ldots,p-2\}\times\{0\}\times\{1\}$, $C=\{1,2,\ldots,p-1\}\times\{1\}\times\{0\}$, $D=\{1,2,\ldots,p-1\}\times\{0\}\times\{1\}$, $E=\{(0,1,0),(p-1,0,0)\}$, and $F=\{(0,1,1),(p-1,0,1)\}$.
This $\cR$ is a coherent and GELO partition of $\indexset$.
Then let $\tau\colon\indexset \to \N$ be defined by $\tau_\alpha=0$ if $\alpha \in A\cup F$ and $\tau_\alpha=1$ if $\alpha \in C \cup D$ and $\tau_\alpha=2$ if $\alpha\in B \cup E$.
It is straightforward to check that $\tau \in \CAs(\cR,=,3) \subseteq \CAs(\cR,=,\ell) \subseteq \CAs(\cR,=)$, which makes $\cR$ coherently satisfiable (and hence coherently contributory since it is GELO), and makes $\card{\CAs(\cR,=,\ell)}$ strictly positive.
This makes $\mu^\ell_{p,(f,g)} \SSCC(f,g)$ strictly positive by \cref{Francis}.
\end{proof}

\section{Variance}\label{Anne}

In this section, we calculate $\cmomv{2} \CDF(f,g)$, that is, the second central moment (variance) of the crosscorrelation demerit factor.
We are working with $\SSCC$ because it is mathematically more convenient than $\CDF$, and since \eqref{Celeste} shows how to convert between $\CDF$ and $\SSCC$, it is simple to convert the results to be in terms of $\CDF$.
We shall be using the formula from \cref{Vito} to compute the variance, so we must first determine the coherent isomorphism classes of coherently contributory partitions.
\begin{lemma}\label{Vivian}
There is one and only one coherent isomorphism class $\fC$ in $\CIsom(2)$.  This sole coherent isomorphism class consists of two partitions, one of which is 
\[
\cP=\left\{\{(0,s,v),(1,s,v)\}: s \in [2], v \in [2]\right\}.
\]
\end{lemma}
\begin{proof}
Note that the partition $\cP$ is just the partition of the same name in \cref{Electra}, which is shown to be a representative of the sole class in $\CIsom(2)$ in \cref{Hector}, and this class is shown to consist of two partitions in \cref{Penelope}.
\end{proof}
Now we determine the solution count function for the unique isomorphism class of coherently contributory partitions.
\begin{lemma}\label{Wilhelm}
For the sole coherent isomorphism class $\fC$ of $\CIsom(2)$, we have
\[
\CSols(\fC,\ell)=\frac{2 \ell^3-3\ell^2+\ell}{3}.
\]
\end{lemma}
\begin{proof}
In \cref{Vivian}, we give a partition $\cP$ that is a representative of the sole coherent isomorphism class of $\CIsom(2)$, and so $\CSols(\fC,\ell)=\card{\CAs(\cP,=,\ell)}$ for all $\ell \in \N$.
Since the matrix
\[
K_1 = \begin{bmatrix}
+1_r & -1_r & +1_b & -1_b \\
+1_r & -1_r & +1_b & -1_b
\end{bmatrix}
\]
is a display matrix for $\cP$, \cref{Timothy} says that $\card{\CAs(\cP,=,\ell)}$ is equal to the number of solutions $x \in [\ell]^4$ of the homogeneous system $K_1 x = 0$ where the first two coordinates of $x$ are colored red, the last two are colored blue, and no two entries of the same color have the same numerical value (proper coloring).
If we write
\[
x=\begin{bmatrix} x_0 \\ x_1 \\ x_2 \\ x_3 \end{bmatrix}
\]
Then $K_1 x=0$ is equivalent to the system with the single equation $x_0+x_2=x_1+x_3$ which by \cite[Lemma C.6(i)]{Katz-Ramirez} has $(2\ell^3+\ell)/3$ solutions in $[\ell]^4$ if we do not insist on proper coloring.
The improperly colored solutions are precisely the ones with $x_0=x_1$ and $x_2=x_3$, of which there are $\ell^2$ total solutions.
So the total number of properly colored solutions is $(2\ell^3+\ell)/3-\ell^2$.
\end{proof}
Now that we have determined the unique isomorphism class of coherently contributory partitions and its solution count function, we can calculate the variance of $\SSCC$.
\begin{theorem}
If $\ell$ is a positive integer, then
\[
\cmomv{2} \SSCC(f,g) = \frac{4\ell^3  - 6\ell^2 + 2\ell}{3}.
\]
\end{theorem}
\begin{proof}
By \cref{Vito}, we have
\[
\cmomv{2} \SSCC(f,g) = \sum_{\fP \in \CIsom(2)} \card{\fP} \CSols(\fP,\ell),
\]
and since \cref{Vivian} tells us that we have only one class $\fC$ in $\CIsom(2)$ with $\card{\fC}=2$, while \cref{Wilhelm} tells us that $\CSols(\fC,\ell)=(2 \ell^3-3\ell^2+\ell)/3$, we obtain the desired result.
\end{proof}
By \eqref{Celeste}, we can divide the variance of $\SSCC$ by $\ell^4$ to get the variance of $\CDF$, and thus obtain \cref{Valerie}.

\section{Skewness}\label{Sidney}

In this section, we calculate $\cmomv{3} \CDF(f,g)$, that is, the third central moment of the crosscorrelation demerit factor, from which one can obtain its skewness by dividing by the $3/2$ power of the variance (calculated in \cref{Anne}).
We are working with $\SSCC$ because it is mathematically more convenient than $\CDF$, and since \eqref{Celeste} shows how to convert between $\CDF$ and $\SSCC$, it is simple to convert the results to be in terms of $\CDF$.
We shall be using the formula from \cref{Vito} to compute the third central moment, so we must first determine the coherent isomorphism classes of coherently contributory partitions.
\begin{lemma}\label{Alexander}
There are precisely two coherent isomorphism classes, $\fC_0$ and $\fC_1$, in $\CIsom(3)$, which respectively contain the partitions
\[
\cP_0=\left\{\{(e,0,v),(\langle e+1\rangle,1,v)\}: e \in [3], v \in [2] \right\},
\]
and
\[
\cP_1=\left\{\{(e,v,v),(\langle e+1 \rangle,1-v,v)\}: e \in [3], v \in [2]\right\},
\]
where $\langle \cdot \rangle \colon \Z \to [3]$ is the function such that $\langle a \rangle$ is the unique element of $[3]$ that is congruent to $a$ modulo $3$.
Furthermore, $\card{\fC_1}=\card{\fC_2}=8$.
\end{lemma}
\begin{proof}
We use our procedure summarized at the end of \cref{Theresa} to find display matrices for a set of representatives of the classes in $\CIsom(3)$.
First we construct the set $A$ of representatives modulo column permutation of all matrices with $3$ rows and entries from $\{0,1\}$ where each row sum is $2$ and each column sum is a positive even integer (necessarily $2$).
This means that the sum of all entries is $6$, and every row and column sum is $2$, and up to column permutation there is only one such matrix, so we let $A=\{L\}$ where 
\[
L=\begin{bmatrix} 0 & 1 & 1 \\ 1 & 0 & 1 \\ 1 & 1 & 0 \end{bmatrix}.
\]

Then we construct the set $S$ of matrices obtained by changing one entry in each row of $L$ from a $+1$ to a $-1$ and we choose representatives modulo column permutation, and so we obtain $S=\{L_0,L_1,\ldots,L_7\}$, where
\begin{center}
\begin{tabular}{ll}
$L_0=\begin{bmatrix}
 0 & -1 & +1 \\
+1 &  0 & -1 \\
-1 & +1 &  0 
\end{bmatrix}$, &
$L_1=\begin{bmatrix}
 0 & +1 & -1 \\
-1 &  0 & +1 \\
+1 & -1 &  0
\end{bmatrix}$, \\[20pt]
$L_2=\begin{bmatrix}
 0 & +1 & -1 \\
+1 &  0 & -1 \\
+1 & -1 &  0
\end{bmatrix}$, &
$L_3=\begin{bmatrix}
 0 & -1 & +1 \\
-1 &  0 & +1 \\
-1 & +1 &  0 
\end{bmatrix}$, \\[20pt]
$L_4=\begin{bmatrix}
 0 & +1 & -1 \\
-1 &  0 & +1 \\
-1 & +1 &  0
\end{bmatrix}$, &
$L_5=\begin{bmatrix}
 0 & -1 & +1 \\
+1 &  0 & -1 \\
+1 & -1 &  0 
\end{bmatrix}$, \\[20pt]
$L_6=\begin{bmatrix}
 0 & -1 & +1 \\
-1 &  0 & +1 \\
+1 & -1 &  0 
\end{bmatrix}$, &
$L_7=\begin{bmatrix}
 0 & +1 & -1 \\
+1 &  0 & -1 \\
-1 & +1 &  0
\end{bmatrix}$.
\end{tabular}
\end{center}

Now we form the set $T$ of all $64$ matrices of the form $[M_0|M_1]$ with $(M_0,M_1) \in S\times S$ and with the nonzero entries in $M_0$ colored red and the nonzero entries in $M_1$ colored blue.
Then we form a set $U$ of representatives of matrices in $T$ modulo the group action that allows us to arbitrarily permute columns, permute rows, negate any subset of rows, and globally swap red and blue colors for all entries of the matrix.
In the rest of the proof, we use the convention that if $A$ and $B$ are two $3\times 3$ matrices, then $[A|B]$ is the $3\times 6$ matrix whose nonzero entries in the three leftmost columns are red and the whose nonzero entries in the three rightmost columns are blue.
Also note that we index rows and columns of matrices starting from $0$.
We may choose our $U$ to be a subset of $\{[L_0|L_j]: j \in [8]\}$ because every matrix in $S$ may be changed to $L_0$ by negating selected rows.
If we begin with $[L_0|L_2]$ and apply row transposition $(12)$ followed by column permutation $(12)(45)$ and negate every row, we obtain $[L_0|L_4]$.
And if we begin with $[L_0|L_2]$ and apply row transposition $(01)$ followed by column permutation $(01)(34)$ and negate every row, we obtain $[L_0|L_6]$.
Since $L_3=-L_2$, $L_5=-L_4$, and $L_7=-L_6$, we see that $[L_0|L_3]$ can be changed into $[L_0|L_5]$ and into $[L_0|L_7]$ using our group operations in the same way we did to get $[L_0|L_4]$ and $[L_0|L_6]$ from $[L_0|L_2]$.
So in fact, we may choose our $U$ to be a subset of $\{[L_0|L_0],[L_0|L_1],[L_0|L_2],[L_0|L_3]\}$.

Note that permuting rows and columns, negating selected rows, and changing colors of entries does not change the rank of a matrix (we ignore color when determining rank).
Now recall that we say that a column of a matrix is {\it red} (resp., {\it blue}) to mean that all its nonzero entries are red (resp., blue), and so according to our conventions for this proof, if $A$ and $B$ are $3\times 3$ matrices, then the first three columns of $[A|B]$ are red and the last three are blue.
We say that a red (resp., blue) column of a matrix is {\it the clone} of a blue (resp., red) column of the same matrix to mean that the only difference between the two columns is the color of the nonzero entries.
We say that a matrix has the {\it clone property} if it has two columns that are clones of each other.
Note that for a matrix of the form $[A|B]$ with $A$ and $B$ being $3\times 3$ matrices, the clone property (or lack of it) is preserved by swapping all colors of nonzero entries, by row and column permutations, and by negation of any selection of rows.
Using the two invariants of rank and clone property, we note that the four matrices $[L_0|L_0]$, $[L_0|L_1]$, $[L_0|L_2]$, and $[L_0|L_3]$ all lie in separate orbits modulo the group action that allows us to arbitrarily permute columns, permute rows, negate any subset of rows, and globally swap red and blue colors for all entries of the matrix.
This is because $[L_0|L_0]$ and $[L_0|L_1]$ have rank $2$ while $[L_0|L_2]$ and $[L_0|L_3]$ have rank $3$, while $[L_0|L_0]$ and $[L_0|L_3]$ have the clone property while $[L_0|L_1]$ and $[L_0|L_2]$ do not.
So in fact, we can take $U=\{M_0,M_1,M_2,M_3\}$ where $M_j=[L_0|L_j]$ for $j\in\{0,1,2,3\}$.

When we compute the reduced row echelon form $M_j'$ for each $M_j \in U$ (with the entries of any given column in $M_j'$ given the same color as the corresponding column of $M_j$), we obtain
\begin{center}
\begin{tabular}{ll}
$M_0' =
\left[
\begin{matrix}
+1_r &  0_r & -1_r \\
 0_r & +1_r & -1_r \\
 0_r &  0_r &  0_r 
\end{matrix}
\,\, \left| \,\,
\begin{matrix}
+1_b &  0_b & -1_b \\ 
 0_b & +1_b & -1_b \\ 
 0_b &  0_b &  0_b
\end{matrix}
\right.
\right]$
&
$M_1' =
\left[
\begin{matrix}
+1_r &  0_r & -1_r \\
 0_r & +1_r & -1_r \\
 0_r &  0_r &  0_r 
\end{matrix}
\,\, \left| \,\,
\begin{matrix}
-1_b &  0_b & +1_b \\ 
 0_b & -1_b & +1_b \\ 
 0_b &  0_b &  0_b
\end{matrix}
\right.
\right]$
\\[20pt]
$M_2'=
\left[
\begin{matrix}
+1_r &  0_r & -1_r \\
 0_r & +1_r & -1_r \\
 0_r &  0_r &  0_r 
\end{matrix}
\,\, \left| \,\,
\begin{matrix}
 0_b &  0_b &  0_b \\
 0_b & -1_b & +1_b \\ 
+1_b &  0_b & -1_b
\end{matrix}
\right.
\right]$
&
$M_3' =
\left[
\begin{matrix}
+1_r &  0_r & -1_r \\
 0_r & +1_r & -1_r \\
 0_r &  0_r &  0_r 
\end{matrix}
\,\, \left| \,\,
\begin{matrix}
 0_b &  0_b &  0_b \\
 0_b & +1_b & -1_b \\ 
+1_b &  0_b & -1_b
\end{matrix}
\right.
\right]$
\end{tabular}
\end{center}
Now we notice that $M_2'$ and $M_3'$ have rows whose colored weight is $(2,0)$.
But $M_0'$ and $M_1'$ have no row nor any difference of distinct rows whole colored weight is $(2,0)$ and $(0,2)$.
Therefore, if we let $D=\{[L_0|L_0],[L_0|L_1]\}$, then $D$ has one and only one display matrix for each class of $\CIsom(3)$, where each display matrix represents a partition from the corresponding class.
Note that $[L_0|L_0]$ represents the partition $\cP_0$ and $[L_0|L_1]$ represents the partition $\cP_1$, so we conclude that $\CIsom(3)$ has precisely two classes with representatives $\cP_0$ and $\cP_1$.

Now we consider the stabilizer of $\cP_0$ under the action of $\Grthree=S_2\times \Wrthree$.  In the following, $\sgn$ maps a permutation to its sign (which is $+1$ for even permutations and $-1$ for odd permutations).  It is easy enough to check that $\Stab_{\Grthree}(\cP_0)=\{(\digamma,((\sigma_0,\sigma_1,\sigma_2),\epsilon))\in\Grthree : \epsilon \in S_3, \digamma \in S_2, \sigma_0=\sigma_1=\sigma_2=(01)^{(1-\sgn(\epsilon))/2}\}$, so $\card{\Stab_{\Grthree}(\cP_0)}=12$ and note that $|\Grthree|=96$, so then the coherent isomorphism class of $\cP_0$ has $96/12=8$ partitions in it.
Likewise, one can check that $\Stab_{\Grthree}(\cP_1)=\{(\digamma,((\sigma_0,\sigma_1,\sigma_2),\epsilon))\in\Grthree : \epsilon \in S_3, \digamma \in S_2, \sigma_0=\sigma_1=\sigma_2=(01)^{(1-\sgn(\epsilon))/2}\digamma\}$, so $\card{\Stab_{\Grthree}(\cP_1)}=12$, so then the coherent isomorphism class of $\cP_1$ also has $96/12=8$ partitions in it.
\end{proof}
Now we determine the solution count function for each isomorphism class of coherently contributory partitions.
\begin{lemma}\label{Bob}
For the coherent isomorphism classes $\fC_0$ and $\fC_1$ of $\CIsom(3)$ described in \cref{Alexander}, we have
\[
\CSols(\fC_0,\ell)=\CSols(\fC_1,\ell)=\frac{\ell^4-4\ell^3+5\ell^2-2\ell}{2}.
\]
\end{lemma}
\begin{proof}
In \cref{Alexander}, we give a partition $\cP_0$ that is a representative of the coherent isomorphism class $\fC_0$, and so $\CSols(\fC_0,\ell)=\card{\CAs(\cP_0,=,\ell)}$ for all $\ell \in \N$.
Since the proof of \cref{Alexander} gives the matrix
\[
[L_0|L_0]=\left[
\begin{matrix}
 0   & -1_r & +1_r \\
+1_r &  0   & -1_r \\
-1_r & +1_r &  0
\end{matrix}
\,\, \left| \,\,
\begin{matrix}
 0   & -1_b & +1_b \\   
+1_b &  0   & -1_b \\   
-1_b & +1_b &  0
\end{matrix}
\right.
\right]
\]
as a display matrix for $\cP_0$, \cref{Timothy} says that $\card{\CAs(\cP_0,=,\ell)}$ is equal to the number of solutions $x \in [\ell]^6$ of the homogeneous system $[L_0|L_0] x = 0$ where the first three coordinates of $x$ are colored red, the last three are colored blue, and no two entries of the same color have the same numerical value (i.e., we have proper coloring).
If we write
\[
x=\begin{bmatrix} x_0 \\ x_1 \\ x_2 \\ x_3 \\ x_4 \\ x_5 \end{bmatrix}
\]
Then $[L_0|L_0] x=0$ is equivalent to the system $x_0+x_3=x_1+x_4=x_2+x_5$.
So we want $x \in [\ell]^6$ satisfying this system with $x_0,x_1,x_2$ distinct (which forces $x_3,x_4,x_5$ to also be distinct).
Since $x_0,\ldots,x_5 \in [\ell]$, the three sums $x_0+x_3$, $x_1+x_4$, and $x_2+x_5$ have a common value $h$ that can range from $0$ to $2(\ell-1)$.
Given a value of $h \in [2\ell-1]$, the number of pairs $(a,b) \in [\ell]^2$ with $a+b=h$ is
\begin{align*}
f(h)=\begin{cases}
h+1 & \text{if $h \leq \ell-1$} \\
2 \ell-1-h & \text{if $h \geq \ell-1$}.
\end{cases}
\end{align*}
For each $h \in [2\ell-1]$, we need to count how many triples $((a_0,b_0),(a_1,b_1),(a_2,b_2))$ of such pairs there are with the three pairs distinct, which is
\[
\sum_{h=0}^{2(\ell-1)} f(h)(f(h)-1)(f(h)-2) = 6 \sum_{h=0}^{2(\ell-1)} \binom{f(h)}{3},
\]
and since $f(h)$ takes values $1,2,\ldots,\ell-1,\ell,\ell-1,\ldots,2,1$ as $h$ runs from $0$ to $2(\ell-1)$, our sum is the same as
\begin{align*}
12 \sum_{j=1}^{\ell-1} \binom{j}{3} + 6 \binom{\ell}{3}
& = 12 \sum_{j=0}^{\ell-1} \binom{j}{3} + 6 \binom{\ell}{3} \\
& = 12 \binom{\ell}{4} + 6 \binom{\ell}{3} \\
& = \frac{\ell(\ell-1)^2(\ell-2)}{2},
\end{align*}
where the second equality uses the analogue of integration in the theory of finite differences (which is essentially repeated application of Pascal's rule).  When one multiplies out the terms in the final expression, one obtains the claimed formula for $\CSols(\fC_0,\ell)$.

In \cref{Alexander}, we give a partition $\cP_1$ that is a representative of the coherent isomorphism class $\fC_1$, and so $\CSols(\fC_1,\ell)=\card{\CAs(\cP_1,=,\ell)}$ for all $\ell \in \N$.
Since the proof of \cref{Alexander} gives the matrix
\[
[L_0|L_1]=\left[
\begin{matrix}
 0   & -1_r & +1_r \\
+1_r &  0   & -1_r \\
-1_r & +1_r &  0
\end{matrix}
\,\, \left| \,\,
\begin{matrix}
 0   & +1_b & -1_b \\ 
-1_b &  0   & +1_b \\ 
+1_b & -1_b &  0      
\end{matrix}
\right.
\right]
\]
as a display matrix for $\cP_1$, \cref{Timothy} says that $\card{\CAs(\cP_1,=,\ell)}$ is equal to the number of solutions $x \in [\ell]^6$ of the homogeneous system $[L_0|L_1] x = 0$ where the first three coordinates of $x$ are colored red, the last three are colored blue, and no two entries of the same color have the same numerical value (proper coloring).
If we write
\[
x=\begin{bmatrix} x_0 \\ x_1 \\ x_2 \\ x_3 \\ x_4 \\ x_5 \end{bmatrix}
\]
Then $[L_0|L_1] x=0$ is equivalent to the system $x_0-x_3=x_1-x_4=x_2-x_5$, and we must find how many solutions this has in $[\ell]^6$ with $x_0,x_1,x_2$ distinct (which forces $x_3,x_4,x_5$ to be distinct).
If we set $y_j=x_j$ for $j \in \{0,1,2\}$ and $y_j=\ell-1-x_j$ for $j \in \{3,4,5\}$, then this is the same as counting how many $y=(y_0,\ldots,y_5) \in [\ell]^6$ with $y_0,y_1,y_2$ distinct make $y_0+y_3=y_1+y_4=y_2+y_5$, which is exactly what we counted when we computed $\CSols(\fC_0,\ell)$, so $\CSols(\fC_1,\ell)=\CSols(\fC_0,\ell)$.
\end{proof}
Now that we have determined the isomorphism classes of coherently contributory partitions and their solution count functions, we can calculate the third central moment of $\SSCC$.
\begin{theorem}
If $\ell$ is a positive integer, then
\[
\cmomv{3} \SSCC(f,g) = 8\ell^4-32\ell^3+40\ell^2-16\ell.
\]
\end{theorem}
\begin{proof}
By \cref{Vito}, we have
\[
\cmomv{3} \SSCC(f,g) = \sum_{\fP \in \CIsom(3)} \card{\fP} \CSols(\fP,\ell),
\]
and since \cref{Alexander} tells us that we have two classes $\fC_0$ and $\fC_1$ in $\CIsom(3)$ with $\card{\fC_0}=\card{\fC_1}=8$, while \cref{Bob} tells us that $\CSols(\fC_0,\ell)=\CSols(\fC_1,\ell)=(\ell^4-4\ell^3+5\ell^2-2\ell)/2$, we obtain the desired result.
\end{proof}
By \eqref{Celeste}, we can divide the third central moment of $\SSCC$ by $\ell^6$ to get the third central moment of $\CDF$, and thus obtain \cref{Scott}.

\section{Computer-assisted calculation of higher moments}\label{Clarence}

We used a computer program to find the fourth through sixth central moments of $\SSCC$.  To determine the $p$th moment, the program first finds one display matrix $M$ for each coherent isomorphism class $\fC$ in $\CIsom(p)$ using the scheme outlined at the end of \cref{Theresa}.
For each such matrix $M$, the program computes the size of the $\Grp$-stabilizer of the partition $\cP$ represented by $M$; from this it computes the size of the $\Grp$-orbit of $\cP$, i.e., it computes $\card{\fC}$.
For each $\fC\in\CIsom(p)$ corresponding to display matrix $M$ and the partition $\cP$ represented by $M$, the program computes $\CSols(\fC,\ell)=\card{\CAs(\cP,=,\ell)}$ (see \cref{Baluga}) by using \cref{Timothy} to convert this to a problem of counting solutions of a homogeneous system that lie within a hypercube and satisfy a constraint that certain coordinates of a solution cannot take equal values.
Ehrhart theory \cite[Ch.~3]{Beck-Robins} is used to count solutions in hypercubes and M\"obius inversion techniques are used to deduct those that conflict with this constraint, so that we obtain $\CSols(\fC,\ell)$ as a quasipolynomial function of $\ell$.
Then the program uses \cref{Vito} to compute the formula for the $p$th central moment of $\SSCC$ as a function of $\ell$ from the sizes of the coherent isomorphism classes $\fC$ and the values of $\CSols(\fC,\ell)$ attached to those classes.

The program verifies our calculations (see Sections \ref{Anne} and \ref{Sidney}) of the $p$th central moment of $\SSCC$ for $p \in \{2$, $3\}$.
From the $p$th central moment of $\SSCC$ computed by the program for $p \in \{4$, $5$, $6\}$, we computed the $p$th central moment of $\CDF$ by dividing by $\ell^{2 p}$ (see \eqref{Celeste}), yielding Theorems \ref{Kurt}---\ref{Methuselah}.

Recall from the hand calculations in Sections \ref{Anne} and \ref{Sidney} that $\CIsom(2)$ has a single coherent isomorphism class of size $2$; that $\CIsom(3)$ has two coherent isomorphism classes, each of size $8$; and that all the counting functions associated with these classes are polynomials (i.e., quasipolynomials of period $1$).
The fourth central moment calculation took a single CPU on a desktop computer about one second to compute, and it found that $\CIsom(4)$ contains $17$ coherent isomorphism classes of partitions, with the set of class sizes being $\{8$, $12$, $48$, $96$, $192$, $384\}$.
The functions $\CSols(\fC,\ell)$ for $\fC\in\CIsom(4)$ are quasipolynomials of periods $1$ and $2$.
The fifth central moment calculation took a single CPU on a desktop computer only a few minutes to compute, and it found that $\CIsom(5)$ contains $80$ coherent isomorphism classes of partitions, with the set of class sizes being $\{160$, $320$, $384$, $768$, $960$, $1920$, $3840$, $7680\}$.
The functions $\CSols(\fC,\ell)$ for $\fC\in\CIsom(5)$ are quasipolynomials with the set of periods being $\{1$, $2$, $3$, $6\}$.
The sixth central moment calculation took a single CPU on a desktop computer about two months, and it found that $\CIsom(6)$ contains $1952$ coherent isomorphism classes of partitions, with the set of class sizes being $\{32$, $120$, $240$, $640$, $960$, $1280$, $1440$, $1920$, $2880$, $3840$, $5760$, $7680$, $11520$, $15360$, $23040$, $30720$, $46080$, $92160\}$.
The functions $\CSols(\fC,\ell)$ for $\fC\in\CIsom(6)$ are quasipolynomials with the set of periods being $\{1$, $2$, $3$, $4$, $5$, $6$, $12$, $20$, $30$, $60\}$.

\appendix

\section{Lemmas used to prove \cref{Francis}} \label{Cat}
This section proves some algebraic and probabilistic lemmas used in the proof of \cref{Francis}.
Throughout this section, if $f=(\ldots,f_0,f_1,f_2,\ldots)$ and $g=(\ldots,g_0,g_1,g_2,\ldots)$ are sequences of elements in a commutative ring, if $h=(f,g)$, if $E\subseteq\N$, and if $\tau\in\As(E)$, then we use the notation $h^{\tau}$ as a shorthand for $\left(\prod_{\alpha \in \leindexset} f_{\tau_\alpha}\right)\left(\prod_{\beta \in \reindexset} g_{\tau_\beta}\right)$.
The first result is essentially a distributive law.
\begin{lemma}\label{Gerald}
Let $k,\ell \in \N$, let $E_1,\ldots,E_k$ be pairwise disjoint subsets of $\N$, and let $E=\bigcup_{j=1}^k E_j$.  Suppose that $R$ is a commutative ring and that $h=(f,g)$ with $f=(\ldots,f_0,f_1,f_2,\ldots)$ and $g=(\ldots,g_0,g_1,g_2,\ldots)$ sequences of elements of $R$.  Then
\[
\sum_{\tau \in \As(E,=,\ell)} h^\tau = \prod_{j=1}^k \,\, \sum_{\tau^{(j)} \in \As(E_j,=,\ell)} h^{\tau^{(j)}}.
\]
\end{lemma}
\begin{proof}
For each $j$ in $\{1,2,\ldots,k\}$, we let $m_j\colon \As(E_j,=,\ell) \to R$ be the map with $m_j(\upsilon)=h^\upsilon$.  Then apply \cite[Lemma A.1]{Katz-Ramirez} with our $m_j$ in place of their $g_j$ to obtain
\begin{align*}
\prod_{j=1}^k \,\, \sum_{\tau^{(j)} \in \As(E_j,=,\ell)} h^{\tau^{(j)}}
& = \sum_{\tau \in \As(E,=,\ell)} \,\, \prod_{j=1}^k m_j(\tau\vert_{E_j})\\
& = \sum_{\tau \in \As(E,=,\ell)} \,\, \prod_{j=1}^k h^{\tau\vert_{E_j}}\\
& = \sum_{\tau \in \As(E,=,\ell)} \,\, \prod_{j=1}^k \left(\left(\prod_{\alpha \in E_j\times [2]\times\{0\}} f_{\tau_\alpha}\right) \left(\prod_{\beta \in E_j\times [2]\times\{1\}} g_{\tau_\beta}\right)\right) \\
& = \sum_{\tau \in \As(E,=,\ell)} \left(\prod_{\alpha \in E \times [2]\times\{0\}} f_{\tau_\alpha}\right) \left(\prod_{\beta \in E \times [2]\times\{1\}} g_{\tau_\beta}\right)  \\
& = \sum_{\tau \in \As(E,=,\ell)} h^\tau. \qedhere
\end{align*}
\end{proof}
The following applies expectations to a distributive law.
\begin{lemma}\label{Hortense}
Let $k,\ell \in \N$, let $E_1,\ldots,E_k$ be pairwise disjoint subsets of $\N$, and let $E=\bigcup_{j=1}^k E_j$.  Then
\[
\sum_{\tau \in \As(E,=,\ell)} \,\, \prod_{j=1}^k \evh(h^{\tau\vert_{E_j}}) = \prod_{j=1}^k \,\, \sum_{\tau^{(j)} \in \As(E_j,=,\ell)} \evh(h^{\tau^{(j)}}).
\]
\end{lemma}
\begin{proof}
For each $j$ in $\{1,2,\ldots,k\}$, we let $m_j \colon \As(E_j,=,\ell) \to \C$ be the map with $m_j(\upsilon)= \evh(h^\upsilon)$.
Then this becomes a special case of \cite[Lemma A.1]{Katz-Ramirez} where one uses our $m_j$ in place of their $g_j$.
\end{proof}
Our final lemma on the expectation value of a product of sequence terms is what allows us to change our probability calculation in the proof of \cref{Francis} to a counting problem.
\begin{lemma}\label{Irene}
Let $p,\ell \in \N$, let $\cP \in \Coh(p)$, let $\tau \in \CAs(\cP,=,\ell)$, and let $E \subseteq [p]$.  Then
\[
\evh(h^{\tau\vert_E}) = \begin{cases}
1 & \text{if $\cP_E$ is even} \\
0 & \text{otherwise.}
\end{cases}
\]
\end{lemma}
\begin{proof}
By our conventions
\begin{align*}
h^{\tau\vert_E}
& = \left(\prod_{\alpha \in E \times [2]\times\{0\}} f_{(\tau\vert_E)_\alpha} \right) \left(\prod_{\beta \in E \times [2]\times\{1\}} g_{(\tau\vert_E)_\beta}\right) \\
& = \left(\prod_{j \in [\ell]} f_j^{\card{(\tau\vert_E)^{-1}(\{j\})\cap(E\times[2]\times\{0\})}} \right) \left(\prod_{k \in [\ell]} g_k^{\card{(\tau\vert_E)^{-1}(\{k\})\cap (E\times[2]\times\{1\})}}\right).
\end{align*}
Since our expectation $\evh$ makes $h=(f,g)$ uniformly distributed over all $2^{2\ell}$ pairs of binary sequences of length $\ell$, we see that $f_0,f_1,\ldots,f_{\ell-1},g_0,g_1,\ldots,g_{\ell-1}$ are independent and uniformly distributed on $\{-1,1\}$.  Thus, $\evh(h^{\tau\vert_E})$ is $1$ if and only if $|(\tau\vert_E)^{-1}(\{j\})\cap(E\times[2]\times\{0\})|$ and $|(\tau\vert_E)^{-1}(\{k\})\cap(E\times[2]\times\{1\})|$ are even for all $j, k \in [\ell]$.
Otherwise  $\evh(h^{\tau\vert_E})$ is $0$.
But since $\tau\in\CAs(\cP)$, \cref{Daphne}\ref{Nancy} shows that $\tau\vert_E\in\CAs(\cP_E)$, so that the classes of $\cP_E$ are the nonempty sets among $|(\tau\vert_E)^{-1}(\{j\})\cap(E\times[2]\times\{0\})|$ and $|(\tau\vert_E)^{-1}(\{k\})\cap(E\times[2]\times\{1\})|$ for $j, k \in [\ell]$.
So $\evh(h^{\tau\vert_E})$ is $1$ if and only if all these classes are of even size, i.e., if and only if $\cP_E$ is an even partition.
\end{proof}

\section*{Acknowledgement}

The authors thank Arhan Canli of Quantavane for an independent verification that detected a typographical error in \cref{Methuselah} in the first preprint version.

\end{document}